\documentclass{article}
\usepackage[left=1in, right=1in, top=1in, bottom=1in]{geometry}
\usepackage{amsmath, amssymb, amsthm}
\usepackage{changes}
\usepackage{xspace}
\usepackage{float}
\usepackage{graphicx}
\usepackage{hyperref}
\usepackage{doi}
\definechangesauthor[name={Josh}, color=blue]{JAG}
\definechangesauthor[name={Jamie}, color=red]{JTF}

\newcommand{\MaxLin}{\ensuremath{\text{Max-2Lin}}\xspace}
\newcommand{\GammaMaxLin}{\ensuremath{\Gamma\text{-}\MaxLin}\xspace}

\newcommand{\HomLoc}{\ensuremath{\text{1-HomLoc}}\xspace}
\newcommand{\CohLoc}{\ensuremath{\text{1-CohoLoc}}\xspace}

\DeclareMathOperator{\GL}{GL}
\DeclareMathOperator{\SL}{SL}

\newcommand{\cc}[1]{\mathsf{#1}}
\newcommand{\Gap}{\cc{Gap}}
\newcommand{\PP}{\mathcal{P}}
\newcommand{\QQ}{\mathcal{Q}}

\theoremstyle{plain}
\newtheorem{theorem}{Theorem}
\newtheorem{lemma}[theorem]{Lemma}
\newtheorem{proposition}[theorem]{Proposition}
\newtheorem{observation}[theorem]{Observation}

\newtheorem*{conjecture}{Conjecture}

\newtheorem{theoremfull}{Theorem}

\newtheorem{lemmafull}[theoremfull]{Lemma}

\theoremstyle{definition}
\newtheorem{definition}[theorem]{Definition}
\newtheorem{example}[theorem]{Example}
\newtheorem*{problem}{Problem}
\newtheorem{question}[theorem]{Open Question}
\newtheorem{remark}[theoremfull]{Remark}

\newcommand{\F}{\mathbb{F}}
\newcommand{\Z}{\mathbb{Z}}
\newcommand{\ZZ}[1]{\Z/{#1}\Z}
\newcommand{\R}{\mathbb{R}}

\newcommand{\C}{\mathbb C}
\newcommand{\nn}{\mathbb N}

\newcommand{\stext}[1]{\ \ \ \ \ \text{(#1)}}

\title{Computational Topology and the Unique Games Conjecture}
\author{Joshua A. Grochow\footnote{Departments of Computer Science and Mathematics, University of Colorado, Boulder, CO, USA, \texttt{jgrochow@colorado.edu}} \ and Jamie Tucker-Foltz\footnote{Amherst College, Amherst, MA, USA, \texttt{jtuckerfoltz19@amherst.edu}}}
\date{\today}

\begin{document}

\maketitle

\begin{abstract}
Covering spaces of graphs have long been useful for studying expanders (as ``graph lifts'') and unique games (as the ``label-extended graph''). In this paper we advocate for the thesis that there is a much deeper relationship between computational topology and the Unique Games Conjecture. Our starting point is Linial's 2005 observation that the only known problems whose inapproximability is equivalent to the Unique Games Conjecture---Unique Games and \MaxLin---are instances of Maximum Section of a Covering Space on graphs. We then observe that the reduction between these two problems (Khot--Kindler--Mossel--O'Donnell, FOCS 2004; SICOMP, 2007) gives a well-defined map of covering spaces. We further prove that inapproximability for Maximum Section of a Covering Space on (cell decompositions of) closed 2-manifolds is also equivalent to the Unique Games Conjecture. This gives the first new ``Unique Games-complete'' problem in over a decade.

Our results partially settle an open question of Chen and Freedman (SODA 2010; Disc. Comput. Geom., 2011) from computational topology, by showing that their question is almost equivalent to the Unique Games Conjecture. (The main difference is that they ask for inapproximability over $\ZZ{2}$, and we show Unique Games-completeness over $\ZZ{k}$ for large $k$.) This equivalence comes from the fact that when the structure group $G$ of the covering space is Abelian---or more generally for principal $G$-bundles---Maximum Section of a $G$-Covering Space is the same as the well-studied problem of 1-Homology Localization.

Although our most technically demanding result is an application of Unique Games to computational topology, we hope that our observations on the topological nature of the Unique Games Conjecture will lead to applications of algebraic topology to the Unique Games Conjecture in the future.
\end{abstract}

\section{Introduction}
A unique game is a constraint satisfaction problem in which every constraint is between two variables, say $x_i$ and $x_j$, and for each assignment to $x_i$, there is a unique assignment to $x_j$ which satisfies the constraint, and vice versa; in particular, the domain of each variable must have the same size $k$. Khot \cite{khot} conjectured that, for any $\varepsilon,\delta > 0$, there is some $k$ such that it is $\cc{NP}$-hard to distinguish between instances of Unique Games in which at most a $\delta$ fraction of constraints can be satisfied from those in which at least a $1-\varepsilon$ fraction of the constraints can be satisfied. The Unique Games Conjecture (UGC) rose to prominence in the past 15 years partly because it implies that our current best approximation algorithms for many problems are optimal assuming $\cc{P} \neq \cc{NP}$ (e.\,g., \cite{khot, khotRegev, chawla, KKMO, khotVishnoi}), thus explaining the lack of further progress on these problems. It is also interesting because, unlike $\cc{P} \neq \cc{NP}$, the UGC is a more well-balanced conjecture, with little consensus in the community as to its truth or falsehood. This even-handedness, together with the progress made in the last 10 years (e.\,g., \cite{AKKSTV, raghavendra, SSE, kolla, MMexpand, AKKT, ABS, khotMoshkovitz}) suggests that the UGC might be closer to resolution than other major conjectures in complexity theory like $\cc{P}$ versus $\cc{NP}$ or $\cc{VP}$ versus $\cc{VNP}$.

Khot, Kindler, Mossel, and O'Donnell \cite{KKMO} showed that the UGC is equivalent to its special case, \GammaMaxLin, in which every constraint is of the form $x_i - x_j = c$, treated as equations over $\ZZ{k}$. This beautiful simplification might lead one to naively expect that the UGC is somehow primarily about linear algebra, but this is potentially misleading. Indeed, a key feature in the solution of linear systems of equations is the ability to perform Gaussian elimination by taking linear combinations of equations, but when the equations are not satisfiable, taking linear combinations of equations can significantly change the maximum fraction of equations that are satisfiable. This leads us to ask: is there a domain of classical mathematics---other than modern computer science---in which the UGC is naturally situated?

In this paper, we argue that (algebraic) topology is such a domain. The starting point for our investigation is Linial's observation \cite{linial}\footnote{\label{fn:linial}In an earlier version of this paper we were unaware of Linial's observation, which appears on slides 55--56 of \cite{linial}. Once we were made aware of this, for which we thank an anonymous reviewer and Hsien-Chih Chang, we checked Linial's slides, and the first author remembered having \emph{attended} the talk that Linial gave at MIT on 11 May 2005! This was, in fact, one of the first theory seminars the first author had ever attended, and at the time he certainly didn't know what bundles were, nor the UGC; he also could not recall whether Linial actually made it to those slides that particular day. We believe that Linial was the first to make this observation.}
that the only two known ``UGC-complete'' problems---UG itself and \GammaMaxLin---are in fact instances of finding a maximum section of a ($G$-)covering space over the underlying constraint graph of the CSP (
topological terminology will be explained in Section~\ref{sec:prelim}; 
$G=S_k$ for UG and $G = \ZZ{k}$ for \GammaMaxLin). In the case of \GammaMaxLin, we observe that this is naturally equivalent to the well-studied 1-Homology Localization problem from computational topology (see, e.\,g., \cite{chenFreedman, CEN, chenFreedman2, EN, chenFreedmanMeasuring, CENFlowsCuts, ZC, DHKUnimodular}). We also observe that the reduction from \GammaMaxLin to UG \cite{KKMO} gives a well-defined map of $G$-covering spaces.


To cement the topological nature of the UGC, we then show that Maximum Section of a $G$-Covering Space, or 1-Homology Localization, on (cell decompositions of) 2-manifolds, rather than graphs, is still UGC-complete. This gives the first new UGC-complete problem in over a decade. Of course, there is some subjectivity as to what counts as a UGC-complete problem being ``distinct'' from UG itself. In particular, $\MaxLin$ can be viewed as UG with certain additional hypotheses satisfied, but nonetheless ``feels'' different (this difference can be made a little more precise topologically, see Remark~\ref{rmk:principal}). In our proof, we'll see that 1-Homology Localization on 2-manifolds can also be viewed as a special case of $\MaxLin$ satisfying certain additional hypotheses, but again, Homology Localization ``feels different'' to us. Regardless, our results draw what we believe is a new connection between UGC and computational topology.
 
The UGC-completeness of this problem also partially settles a question of Chen and Freedman \cite{chenFreedman} on the complexity of the 1-Homology Localization problem on 2-manifolds. In particular, Chen and Freedman \cite[p.~438, just before Section~4.3]{chenFreedman} asked whether it was hard to approximate 1-Homology Localization with coefficients in $\ZZ{2}$ on 2-manifolds; while some of the details are left unspecified, given the context in their paper we may conservatively infer (see our discussion in Section~\ref{sec:CF}) that they were asking for inapproximability to within all constant factors for triangulations of 2-manifolds. We show:
\begin{itemize}
\item Assuming UGC, for any constant $\alpha > 1$, there is a $k$ such that 1-Homology Localization over $\ZZ{k}$ on cell decompositions of 2-manifolds cannot be efficiently approximated to within $\alpha$. In particular, this problem is UGC-complete.

\item Assuming UGC, for any $\varepsilon > 0$, there is a $k$ such that 1-Homology Localization over $\ZZ{k}$ on \emph{triangulations} of 2-manifolds cannot be efficiently approximated to within $7/6 - \varepsilon$.
\end{itemize}

See Section~\ref{sec:future} for questions we leave open.

Although the above are our most technically demanding results, which are applications UGC to 1-Homology Localization---and, in the course of this, showing a new UGC-complete problem---we hope that the connections we have drawn between UGC and computational topology will lead to further progress on both of these topics in the future.


\subsection{Related work} \label{sec:related}
Linial \cite{linial} first observed that UG could be phrased in terms of Maximum Section of a Graph Lift (though we were unaware of this when we began our investigations, see Footnote~\ref{fn:linial}). To our knowledge, since Linial's observation there have been no other works relating approximation problems in computational topology with the Unique Games Conjecture, nor is there previous work on the problem of Maximum Section of $G$-Covering Spaces. In this paper we extend Linial's observation by showing that the reduction of \cite{KKMO} gives a well-defined map of covering spaces, and we relate UGC to the well-studied problem of Homology Localization.

Here we briefly survey related work on approximation problems in computational topology, particularly those related to Homology Localization and the question of Chen and Freedman that we partially answer. Note that $d$-Homology Localization fixes the dimension of the homology considered, but allows the input to consist of $d$-homology classes on manifolds of arbitrary large dimension. In our paper we consider $1$-Homology Localization on graphs and 2-manifolds. For a more comprehensive overview of the area, as well as more direct motivations for the problem of Homology Localization, see \cite[Sections~1 and 3]{chenFreedman}.

1-Homology Localization with coefficients in $\ZZ{2}$ is $\cc{NP}$-hard to optimize exactly on simplicial complexes \cite{chenFreedman2} and even on 2-manifolds \cite{CEN}. Chen and Freedman showed it was $\cc{NP}$-hard to approximate 1-Homology Localization on triangulations of 3-manifolds to within all constant factors, and that it was $\cc{NP}$-hard to approximate $d$-Homology Localization on triangulations of manifolds for any $k \geq 2$. The best known algorithms for 1-Homology Localization on a 2-manifold are given in several papers by Chambers, Erickson, and Nayyeri \cite{CEN,EN,CENFlowsCuts} (see also \cite{ZC}). In particular, \cite{CEN} solve the problem in polynomial time for fixed genus; however, for triangulations of 2-manifolds, the genus $g = \Theta(e/v)$ (follows from Euler's formula), and for instances of Unique Games to be difficult one must have the edge density $e/v$ growing without bound, so their result seems not to solve the instances of 1-Homology Localization relevant for the UGC. Dey, Hirani, and Krishnamoorthy \cite{DHKUnimodular} showed that Homology Localization \emph{in the 1-norm} over $\Z$ can be done in polynomial time; in our paper we are primarily concerned with the 0-norm.

\subsection{Organization}
In Section~\ref{sec:prelim} we give preliminaries, including an introduction to all topological concepts used. 
Section~\ref{sec:complete} contains the details of how to view Unique Games and \GammaMaxLin as instances of Maximum Section of a $G$-Covering Space, and the proof that the KKMO reduction \cite{KKMO} gives a well-defined map of $G$-covering spaces. In Section~\ref{sec:manifolds} we show that 1-Homology Localization on cell decompositions of 2-manifolds is UGC-complete. In Section~\ref{sec:CF} we show how our techniques partially settle a question of Chen and Freedman, and in Section~\ref{sec:future} we discuss questions left open by our results, or suggested by our topological viewpoint. In Appendix~\ref{app:nonab} we discuss generalizations to non-Abelian groups $G$ and arbitrary topological spaces $X$.

A version of this paper appeared in the SoCG 2018 conference proceedings \cite{GTconf}. Although we prefer section-numbered items such as ``Theorem 3.7'', we instead chose to be consistent with the numbering scheme used in the conference proceedings: items that appear in both versions have the same numbers. Appendix~\ref{app:nonab} does not appear in the conference version, though it is mentioned there. Other items that appear only in this full version are either un-numbered or numbered beginning with ``F'' for ``full.'' 

\section{Preliminaries} \label{sec:prelim}
Partially because this paper has at least two natural audiences---those interested in UGC, some of whom may be less familiar with algebraic topology, and vice versa---the preliminaries are a bit longer than usual. 

\subsection{The Unique Games Conjecture and inapproximability} 
We refer to the textbooks \cite{vazirani, aroraBarak} for standard definitions and results on approximation algorithms and inapproximability, and to the survey \cite{khotSurvey} for more on the Unique Games Conjecture. Here we briefly spell out the needed definitions and one standard lemma that will be of use.

An instance of a constraint satisfaction problem (CSP) is specified by a set of variables $x_1, \dotsc, x_n$, for each variable $x_i$ a domain $D_i$ (which we will always take to be a finite set, and, in fact, we will have all $D_i$ equal to one another), and a set of constraints. Each constraint is specified by a subset $\{x_{i_1}, \dotsc, x_{i_k}\}$ of $k$ of the variables (each constraint may, in principle, have a different arity $k$), and a $k$-ary relation $R \subseteq D_{i_1} \times \dotsb \times D_{i_k}$. An assignment to the variables satisfies a given constraint if the assignment is an element of the associated $k$-ary relation. 

A CSP may be specified by restricting the arity and type of relations allowed in its instances, as well as the allowed domains for the variables. The value function associated to a CSP is $v(x,s) = $ the fraction of constraints in $x$ satisfied by $s$, and we get the associated maximization problem. Given a CSP $\PP$, the associated \emph{gap problem} $\Gap \PP_{c, s}$ is the promise problem of deciding, given an instance $x$, whether $OPT(x) \leq s$ or $OPT(x) \geq c$. (An algorithm solving $\Gap \PP_{c, s}$ may make either output 
if $x$ violates the promise, that is, if $s < OPT(x) < c$.) In general, the parameters $c, s$ may depend on the problem size $|x|$.
 
If the optimization problem $\PP$ can be approximated to within a factor $\alpha$ by some algorithm, then essentially the same algorithm solves $\Gap \PP_{c,s}$ whenever $c/s > \alpha$. In the contrapositive, if $\Gap \PP_{c,s}$ is, for example, $\cc{NP}$-hard, then so is approximating $\PP$ to within a factor $c/s$. The converse is false.

\begin{problem}[Unique Games]\label{UGProblem}
The Unique Games problem with $k$ colors, denoted $UG(k)$, is the CSP whose domains all have size exactly $k$, and where each constraint has arity 2 and is a bijection between the domains of its two variables.
\end{problem}

The natural $n$-vertex graph associated to a UG instance---in which there is an edge $(i,j)$ for each constraint on the pair $(x_i, x_j)$---is called its \emph{constraint graph}. A $UG(k)$ instance is completely specified by its constraint graph, together with a a permutation $\pi_{ij} \in S_k$ on each edge $(i,j)$, specifying the constraint that, for $i < j$, $x_i = \pi_{ij}(x_j)$. 

\begin{conjecture}[Khot \cite{khot}, Unique Games Conjecture (UGC)]
	For all $\varepsilon, \delta > 0$, there exists a $k \in \nn$ such that $\Gap UG(k)_{1 - \varepsilon, \delta}$ is $\cc{NP}$-hard.
\end{conjecture}

Since the community is divided on this exact conjecture,  
the UGC is sometimes interpreted more liberally as saying that $\Gap UG(k)_{1 - \varepsilon, \delta}$ is \emph{somehow} hard, for example, not in $\cc{P}$, $\cc{BPP}$, or $\cc{quasiP}$. All our results will work equally well under any of these interpretations, so we often just refer to ``efficient approximation'' or write ``approximating ... is hard.''

A polynomial-time \emph{gap-preserving reduction} from $\Gap \PP_{c,s}$ to $\Gap \QQ_{c',s'}$ (say, both minimization or both maximization problems) is a polynomial-time function $f$ such that $OPT_{\PP}(x) \leq s \Rightarrow OPT_{\QQ}(f(x)) \leq s'$ and $OPT_\PP(x) \geq c \Rightarrow OPT_\QQ(f(x)) \geq c'$. If $\PP$ is a maximization problem and $\QQ$ is a minimization problem, then a gap-preserving reduction is similarly an $f$ such that $OPT_\PP(x) \leq s \Rightarrow OPT_\QQ(f(x)) \geq c'$ and $OPT_\PP(x) \geq c \Rightarrow OPT_\QQ(f(x)) \leq s'$. 

We say informally that a problem $\PP$ is ``UGC-complete'' if there are gap-preserving reductions from $\Gap\PP_{\alpha,\beta}$ to GapUG$_{1-\varepsilon,\delta}$ and GapUG$_{1-\varepsilon,\delta}$ to $\Gap\PP_{\alpha,\beta}$ (where, in one direction, $\varepsilon,\delta$ may depend on $\alpha,\beta$, and vice versa in the other direction) such that some UGC-like statement holds for $\PP$---such as ``For any $\alpha < 1$, $\beta > 0$ $\Gap\PP_{\alpha,\beta}$ is hard to approximate''---if and only if UGC holds. 
Prior to this paper, the only known UGC-complete problems were UG itself,\footnote{And slight variants, for example UG on bipartite graphs \cite{khot}, or a variant due to Khot and Regev \cite{khotRegev} in which one tries to maximize the number of vertices in an induced subgraph all of whose constraints are satisfied, rather than simply maximizing the number of constraints satisfied.} and $\GammaMaxLin(q)$ \cite{KKMO}:

\begin{problem}[$\MaxLin(A)$ and $\GammaMaxLin(A)$]
Let $A$ be an Abelian group. $\MaxLin(A)$, or $\MaxLin(k)$ when $A=\ZZ{k}$, consists of those instances of UG where every variable has $A$ as its domain, and each constraint takes the form $a x_i + b x_j = c$ for some $a,b \in \Z$ and $c \in A$ (not necessarily the same $a,b,c$ for all constraints). $\GammaMaxLin(A)$ is the same, except that all the constraints have the form $x_i - x_j = c$ for some $c \in A$ (not necessarily the same for all constraints).
\end{problem}



We will use the following standard lemma, which allows one to add a small number of new constraints to a given graph in a way that preserves an inapproximability gap.

\begin{lemma}\label{LinEdges}
For a class $\mathcal{A}$ of graphs, let $UG_{\mathcal{A}}$ denote the Unique Games Problem on graphs from $\mathcal{A}$. Given two classes of graphs $\mathcal{A}, \mathcal{B}$, let $f: \mathcal{A} \to \mathcal{B}$ be a polynomial-time computable function such that for all $G \in \mathcal{A}$, $E(G) \subseteq E(f(G))$ and $|E(f(G)) \setminus E(G)| = O(v)$ where $v$ is the number of vertices in $G$ of degree $\geq 1$. If the number of edges added is at most $av$, then there is a gap-preserving reduction from $UG_{\mathcal{A},1-\varepsilon,\delta}$ to $UG_{\mathcal{B}, 1-\varepsilon_0, \delta_0}$ where $\varepsilon_0 = \varepsilon + \Delta$ and $\delta_0 = \delta + \Delta$, for any $1 > \Delta > 2\delta a / (1 + 2\delta a)$ (in particular, with $\Delta \to 0$ as $\delta \to 0$).

In particular, if $UG_\mathcal{A}$ is UGC-hard, then so is $UG_\mathcal{B}$. The same holds with ``UG'' everywhere replaced by $\MaxLin$ or $\GammaMaxLin$.
\end{lemma}

The intuition here is that one can always satisfy a number of constraints linear in the number of vertices (just choose a spanning tree or forest), so adding another linear number of constraints will only affect the inapproximability gap by a constant, which is negligible as $\delta$ and $\varepsilon$ get arbitrarily small. We include its (easy) proof 
for completeness, as we could not find an explicit reference.

\begin{proof}
We will show that for any $\varepsilon_0, \delta_0 > 0$, that there is a polynomial-time many-one reduction from $\Gap UG_{\mathcal{A},1-\varepsilon,\delta}$ to $\Gap UG_{\mathcal{B}, 1-\varepsilon_0, \delta_0}$ for some $\varepsilon,\delta > 0$. In particular, if $UG_{\mathcal{A}}$ is UGC-hard, then so is $UG_{\mathcal{B}}$. 
	
	By assumption, the number of edges added by the graph transformation $f$ is $O(v)$, so suppose that no more than $av$ edges are added, where $a$ is a fixed constant. Choose $\varepsilon, \delta, x > 0$ such that $\varepsilon + x < \varepsilon_0$, $\delta + x < \delta_0$, $x < 1$, and $\delta < \frac{x}{2a(1 - x)}$ (noting that the derivative of $x + \frac{x}{2a(1 - x)}$ is always positive on $[0,1]$ we see that this is always possible). We claim that the following algorithm correctly distinguishes between $\mathcal{A}$-instances that are  $\geq 1-\varepsilon$ satisfiable and those in which at most a $\delta$ fraction of constraints can be satisfied. Given an instance $G$ with $v$ vertices of degree $\geq 1$ and $e$ edges (where both $e$ and $v$ are at least 2), first check if $\frac{v}{e} \geq \frac{x}{a(1 - x)}$. If this is true (case 1), accept. Otherwise (case 2), output the answer for $\Gap UG_{\mathcal{B}, 1-\varepsilon_0, \delta_0}$ on the instance $f(G)$.
	
	
	Clearly this algorithm runs in polynomial time. We will prove that it is correct by showing that $f$ maps highly satisfiable instances to highly satisfiable instances, and highly unsatisfiable instances to highly unsatisfiable instances. Note that it is always possible to satisfy $\lceil\frac{v}{2}\rceil$ edges: iteratively choose an unused vertex and an adjoining edge to satisfy, by choosing a suitable label. (One can potentially do better by choosing a spanning forest and satisfying all the edges in the forest, but we won't need that.)  
	At each step, at most two vertices get used, so we will definitely be able to iterate at least $\lceil\frac{v}{2}\rceil$ times. Therefore, if we are in case 1, we know we can satisfy a
	$$\frac{v}{2e} \geq \frac{x}{2a(1 - x)} > \delta$$
	fraction of edges, so we cannot be in the highly unsatisfiable case, and therefore the algorithm correctly accepts. 
	
	On the other hand, if we are in case 2, then we have
	$$\frac{v}{e} < \frac{x}{a(1 - x)} \implies av < \frac{ex}{(1 - x)}.$$
	Suppose first that $G$ was highly satisfiable, i.\,e., a $(1 - \varepsilon)$ fraction of edges could be satisfied. Because $f$ does not remove edges, those same edges can still be satisfied in $f(G)$. However, the total number of edges may have increased by as much as $av$, so the fraction of edges satisfied may be less. At worst, the fraction of edges satisfied will still be at least
	\begin{align*}
	\frac{(1 - \varepsilon)e}{e + av} &> \frac{(1 - \varepsilon)e}{e + \frac{ex}{1 - x}} \stext{from the previous inequality}\\
	&= \frac{(1 - \varepsilon)}{1 + \frac{x}{1 - x}}\\
	&= \frac{(1 - \varepsilon)}{\frac{1}{1 - x}}\\
	&= (1 - \varepsilon)(1 - x)\\
	&= (1 - x) - \varepsilon(1 - x)\\
	&> 1 - x - \varepsilon \stext{because $x < 1$}\\
	&= 1 - (x + \varepsilon)\\
	&> 1 - \varepsilon_0.
	\end{align*}
	Thus, $f(G)$ is highly satisfiable.
	
	Conversely, suppose $G$ was highly unsatisfiable, i.\,e., less than a $\delta$ fraction of edges can be satisfied. It will still not be possible to satisfy more than that in the original graph when edges are added, but the new edges might be satisfiable. Thus, the fraction of satisfiable edges in $f(G)$ is strictly bounded by
	\begin{align*}
	\frac{\delta e + av}{e + av} &< \frac{\delta e + \frac{ex}{1 - x}}{e + \frac{ex}{1 - x}}\\
	&= \frac{\delta + \frac{x}{1 - x}}{1 + \frac{x}{1 - x}}\\
	&= \frac{\delta + \frac{x}{1 - x}}{\frac{1}{1 - x}}\\
	&= (1 - x)\delta + x\\
	&= x + \delta - x\delta\\
	&< x + \delta\\
	&< \delta_0
	\end{align*}
	as desired.
\end{proof}

\subsection{\texorpdfstring{$G$-covering}{G-covering} spaces of graphs}
\begin{definition}[Graph lifts, a.k.a covering graph]
Let $X$ be a graph. A \emph{graph lift}, or \emph{covering graph}, is another graph $Y$ with a map $p\colon V(Y) \to V(X)$ that such that the restriction of $f$ to the neighborhood of each $v \in V(Y)$ is a bijection onto the neighborhood of $f(v) \in X$. If $X$ is connected, then the number $k$ of points in $f^{-1}(v)$ is independent of $v$, and we say $Y$ is a \emph{$k$-sheeted cover} of $X$. The set of vertices $p^{-1}(v)$ is called the \emph{fiber} over $v$.
\end{definition}

Graph lifts have found many uses in computer science and mathematics, particularly in the study of expanders (e.\,g., \cite{BL, HLW, ACKM}) and, via the next example, Unique Games.

\begin{example}[Label-extended graph] \label{ex:labelExtended}
Given an instance of Unique Games with constraint graph $X$ on vertex set $[n]$ and domain size $k$, and permutations $\pi_e$ on the directed edges $e \in E(X)$, its \emph{label-extended graph} is a graph with vertex set $[n] \times [k]$, and with an edge from $(v,i)$ to $(w,j)$ iff $\pi_{v,w}(i)=j$. In particular, the label-extended graph is a $k$-sheeted covering graph of $X$.
\end{example}

In our setting, all of our covering graphs will come naturally with a group that acts on their fibers, and we would like to keep track of this group action, for reasons that will become clear in Section~\ref{sec:complete}. For example, the label-extended graph of a UG instance carries a natural action of $S_k$ on each fiber (as would be the case with \emph{any} $k$-sheeted covering graph), and the label-extended graph of a $\MaxLin(A)$ instance has a natural action of the Abelian group $A$ on each fiber. From the point of view of approximation, keeping track of the group currently seems of little relevance, but it may be useful from the topological point of view, so we state our definitions and results carefully keeping track of the (monodromy) group.

\newcommand{\covfootnote}{\label{fn:terminology}Note that here we consider $G$ \emph{as a permutation group}---that is, technically, an abstract group \emph{together} with an action on a set of size $k$, as is done in Definition 12 of the preprint \cite{ACKMarXiv}. In \cite[Definition~1]{ACKM} (and \cite[Definition~1]{ACKMarXiv}) they define a ``$G$-lift'' for an abstract group $G$ as a $G$-covering graph where the action of $G$ is the regular action on itself by left translations. To translate between this terminology, that of bundles, and that of voltage graphs \cite{grossTucker}, we have: 
\begin{tabular}{lp{1in}p{1in}lll}
Action & Covering graph & Bundle & Lift & Voltage Graph \\ \hline
regular & regular covering space & principal $G$-bundle & $G$-lift & ordinary voltage graph \\
general & $G$-covering graph (not nec. regular) & general $G$-bundle with finite fibers & $(G,S,\cdot)$-lift & permutation voltage graph
\end{tabular}}

\begin{definition}[{$G$-covering graph, see \cite{grossTucker} and \cite[Definition~12]{ACKMarXiv}\footnote{\covfootnote}}] \label{def:GcovGraph}
Let $G$ be a group of permutations on a set of size $k$. A $G$-covering space of a graph $X$ is a $k$-sheeted covering graph $Z = (V(X) \times [k], E)$ such that the permutations on each edge come from the action of the group $G$. Symbolically, for each edge $(u,v) \in X$, there is a group element $g_{u,v} \in G$ and $Z$ contains an edge from $(u,i)$ to $(v,j)$ iff $g_{u,v}(i)=j$.
\end{definition}

In topological terminology, this definition is equivalent to a ``$G$-bundle with finite fibers'' or to a covering space of the graph whose monodromy group (the group generated by considering the permutations you get by going around cycles in the graph) is contained in $G$.

We consider a graph $X$ as a 1-dimensional geometric simplicial complex in the natural way, in which each edge has length 1.

\begin{definition}[Section of a covering graph]
Given a covering graph $p\colon Y \to X$, a \emph{section} of $p$ is a continuous map $s\colon X \to Y$ (of topological spaces, as above) such that $p(s(x)) = x$ for all $x$. That is, it is a choice, for each $x \in X$, of a unique point in $p^{-1}(x)$, in a way that varies continuously with $x$.
\end{definition}

\begin{example} \label{ex:mobiusDiscrete}
Consider a covering graph $p\colon Y \to X$ where $X$ is a triangle ($V(X) = \{0,1,2\}$, $E(X) = \{\{1,2\}, \{0,1\}, \{0,2\}\}$), $Y$ is a 6-cycle with vertex set $\{0,\dotsc,5\}$, in its natural ordering (edge set $\{\{i,i+1 \pmod{6}\} : i \in \{0,\dotsc,5\}\}$), and $p(i) = i \pmod{3}$. 
This covering graph has no section: For, without loss of generality, we may suppose it has a section $s$ and $s(0)=0$. Then by continuity (imagine dragging a point $x \in Z$ from the point $0$ across the edges of $G$), $s(1)=1$ and $s(2)=2$. But then as we continue varying our point in $X$ across the edge $\{2,0\}$, we find that we must also have $s(0)=3$, a contradiction. If we think of $Y$ as ``lying over'' $X$ in the manner specified by $p$, we see that it is the label-extended graph of the UG instance on $X$ with domain size $2$, and $\pi_e$ being the unique transposition for every edge $e$. The fact that there is no section here corresponds precisely to the fact that the UG instance is not completely satisfiable; see Observation~\ref{obs:UG}.
\end{example}

Given two covering graphs $p_i\colon Y_i \to X$ ($i=1,2$) of the same graph $X$, a \emph{homomorphism} between covering graphs is a continuous map $f\colon Y_1 \to Y_2$ such that $p_2 \circ f = p_1$. In particular, this means that the points in $Y_1$ in the fiber over $x \in X$ (that is, in $p_1^{-1}(x)$) are mapped to points in $Y_2$ that also lie in the fiber over the same point $x$. 

\begin{example}[Isomorphism of label-extended graphs]
Given two instances of UG on the same constraint graph $X$, if their label-extended graphs are isomorphic as covering graphs of $X$, then there is a natural bijection between assignments to the variables in the two instances which precisely preserves the number of satisfied constraints. Indeed, such an isomorphism is nothing more than re-labeling the elements of the domain of each variable.
\end{example}

\begin{observation} \label{obs:isoGcov}
Given two $G$-covering graphs $p_\ell\colon Y_\ell \to X$ ($\ell=1,2$) with edge permutations $\pi_{ij}^{(\ell)}$, any isomorphism of covering graphs between them has the following form: for each $i \in V(X)$ there is a permutation $\pi_i$ (not necessarily in $G$) such that $\pi_{ij}^{(2)} = \pi_i^{-1} \pi_{ij}^{(1)} \pi_j$.

Conversely, given a $G$-covering space $p_1 \colon Y_1 \to X$ with edge permutations $\pi_{ij}$, and an element $g_i \in G$ for each $i \in V(X)$, the $G$-covering space $Y_2$ defined by $\pi_{ij}^{(2)} = g_i \pi_{ij}^{(1)} g_j^{-1}$ is isomorphic to $Y_1$. 
\end{observation}

\begin{definition} \label{def:isoGcov}
An \emph{isomorphism of $G$-covering graphs} $p_\ell \colon Y_\ell \to X$ ($\ell=1,2$) is an isomorphism of covering graphs such that the $\pi_i$ (notation from Observation~\ref{obs:isoGcov}) can be chosen to lie in $G$. When we say two $G$-covering spaces are ``isomorphic'', we mean isomorphic as $G$-covering spaces (not just as covering spaces).
\end{definition}

All these notions generalize from graphs to topological spaces, but we postpone this generalization until Appendix~\ref{app:nonab}.

\subsection{Combinatorial models of topological spaces: CW-, \texorpdfstring{$\Delta$-}{Delta-} and simplicial complexes} \label{app:complexes}
For algebraic topology, we refer the reader to Hatcher's excellent textbook \cite{hatcher}.

A \emph{$d$-simplex} is the topological space which is the intersection of the hyperplane in $\R^{d+1}$ defined by $\sum_{i \in [d]} x_i = 1$ with the positive orthant ($x_i \geq 0$ for all $i$). In particular, a 0-simplex is a point, a 1-simplex is a line, and a 2-simplex is a triangle. We consider the $d$-simplex to be naturally oriented in the order $(1,0,0,\dotsc,0), (0,1,0,\dotsc,0), \dotsc, (0,0,\dotsc,0,1)$. The boundary of such a simplex $\Delta$, denoted $\partial \Delta$, is the set of points where at least one of the coordinates $x_i$ is zero.

A \emph{CW complex} is a topological space $X$ built from simplices in an iterative manner as follows: let $X_{-1}$ be the empty set, and let $X_i$ be the CW complex that has been constructed through stage $i$. At stage $i+1$, $X_{i+1}$ is obtained from $X_i$ by adding a single simplex $\Delta$ to $X_i$, where $\Delta$ is attached to $X_i$ by any continuous map $f\colon \partial \Delta \to X_i$. 
More formally, we define $X_{i+1}$ by considering the disjoint union $X_i \sqcup \Delta$ and then modding out by the equivalence relation $x \sim f(x)$ for all $x$ in the boundary $\partial \Delta$. To be a CW complex, the only further restriction is that we must add our simplices in non-decreasing order of dimension: first all the 0-simplices, then all the 1-simplices, and so on. On the one hand, CW complexes are very flexible because they allow \emph{any} continuous map as the gluing map $\partial \Delta \to X_i$, but this can also make them more difficult to work with. There are several strengthenings of this notion that give up some flexibility for the benefit of more combinatorial definitions.

A \emph{$\Delta$-complex} is a CW complex in which all of the attaching maps $f\colon \partial \Delta \to X_i$ are completely determined by what they do to the vertices of $\Delta$. In particular, if $\dim \Delta = d$, then a $(d-1)$-face of $\Delta$ must get mapped bijectively to an existing $(d-1)$-simplex in $X_i$. A \emph{simplicial complex} is a $\Delta$-complex in which, furthermore, different faces of $\Delta$ must get mapped to distinct simplices of $X_i$, and one cannot add a simplex whose vertex set agrees entirely with a simplex already in $X_i$---in other words, simplices in a simplicial complex are uniquely determined by their vertex set. In particular, simple undirected graphs are essentially the same as simplicial complexes of dimension 1.

While we will have occasion to consider both CW complexes and simplicial complexes, we will restrict all our CW complexes to be ``combinatorial,'' in the following sense. Since we will only consider complexes of dimension at most 2, and the generalization to higher dimensions is somewhat tedious, we only define combinatorial CW complexes of dimension $\leq 2$. In dimension 0 there is no additional restriction---0-dimensional complexes are always just disjoint unions of points. In dimension 1, a combinatorial CW complex is the same as a $\Delta$-complex, which is essentially just an undirected graph with multi-edges and self-loops allowed. In dimension 2, the restriction is that the gluing maps $f\colon \partial \Delta \to X_i$ must map the boundary to a walk along the 1-simplices used in building $X_i$. More precisely, suppose that $e_1, \dotsc, e_n$ are all of the 1-simplices that were attached in the construction of $X_i$. When attaching a 2-simplex $\Delta$, consider the closed loop $t \colon S^1 \to \Delta$ which maps the circle onto the boundary of $\Delta$ in the natural way (each third of the circle gets mapped to one edge of the triangle). It must be the case that there is some $m$ and a walk $e_{i_1}, e_{i_2}, \dotsc, e_{i_m}$ such that the boundary of $\Delta$ maps onto this walk in the natural manner, namely that for each $j \in \{0, \dotsc, m-1\}$, $f(t([\frac{2\pi j}{m}, \frac{2\pi(j+1)}{m}])) = e_{i_j}$. We refer to this walk as being ``defined by the boundary of $\Delta$.''

\subsection{Homology and cohomology in dimensions \texorpdfstring{$\leq 2$}{at most 2}}
Let us briefly recall the problem of 1-Homology Localization, specialized to our context. Given a simplicial complex, or more generally a combinatorial CW complex $X$ (see Section~\ref{app:complexes}) of dimension $2$, the group of $d$-cycles ($d=0,1,2$) with coefficients in an Abelian group $A$, denoted $C_d(X, A)$ is isomorphic to the group $A^{n_d}$, where $n_d$ is the number of $d$-simplices in $X$ ($d=0$: vertices, $d=1$: edges, $d=2$: triangles or 2-cells). We identify the coordinates of such a vector with an assignment of an element of $A$ to each $d$-simplex of $X$. The support of a $d$-chain $a \in C_d(X,A)$ is the set of $d$-simplices that appear in $a$ with nonzero coefficient. The boundary of a $1$-simplex $[i,j]$ is the $0$-chain $\partial_1([i,j]) := [i]-[j]$, and this operator $\partial_1$ is extended to a function $C_1(X,A) \to C_0(X, A)$ by $A$-linearity. Similarly, the boundary of a $2$-cell $[i_1,i_2, \dotsc, i_\ell]$ is the $1$-cycle $\partial_2 [i_1,i_2, \dotsc, i_\ell] := [i_1,i_2] + [i_2,i_3] + [i_3,i_4] + \dotsb + [i_{\ell-1},i_\ell] - [i_1,\ell]$, and we extend this to a map $C_2(X, A) \to C_1(X, A)$ by $A$-linearity. When no confusion may arise, we may refer to both of these maps simply as $\partial$. The image of the boundary map $\partial_d$ is a subgroup of $C_{d-1}(X, A)$, denoted $B_{d-1}(X, A)$.

A \emph{$d$-cycle} is a $d$-chain $a \in C_d(X, A)$ such that $\partial a = 0$. For example, if $X$ is a graph, a 1-cycle is just a union of cycles, in the usual sense of cycles in a graph; if $X$ is a 2-manifold, the only 2-cycles are $A$-scalar multiples of the entire manifold; all vertices are $0$-cycles. The $d$-cycles form a subgroup of $C_d(X, A)$ denoted $Z_d(X, A)$. 

Two $d$-cycles that differ by the boundary of a $(d+1)$-cycle are \emph{homologous}. The $d$-homology classes form the quotient group $H_d(X, A) := Z_d(X, A) / B_d(X, A)$. $H_0(X,A) \cong A^{c}$ where $c$ is the number of connected components, and if $X$ is a closed 2-manifold then $H_2(X, A) = A$. For closed 2-manifolds, thus the main interest is in $H_1(X, A)$.

\begin{problem}[1-Homology Localization, \HomLoc]
Given a simplicial complex (resp., combinatorial CW complex) $X$ and a 1-cycle $a \in Z_1(X, A)$, determine the sparsest homologous representative of $a$, that is, a 1-cycle $a'$ homologous to $a$ with minimum support among all 1-cycles homologous to $a$.
\end{problem}

Cohomology is, in a sense, dual to homology. A $d$-cochain on $X$ with coefficients in an Abelian group $A$ is a homomorphism $C_d(X, \Z) \to A$; equivalently, it is determined by its values (from $A$) on the $d$-simplices (or $d$-cells) of $X$. The $d$-cochains form a group $C^d(X, A) \cong A^{n_d}$, where $n_d$ is the number of $d$-simplices or $d$-cells. Given a $d$-cochain $f\colon X_d \to A$ ($X_d$ being the $d$-simplicies or $d$-cells of $X$), its \emph{coboundary} is the function $(\delta f)\colon X_{d+1} \to A$ defined by $(\delta f)(\Delta) = f(\partial \Delta)$ for any $(d+1)$-cell $\Delta$, and then extended $A$-linearly. Thus $\delta f \in C^{d+1}(X, A)$. A \emph{$d$-cocycle} is a $d$-cochain whose coboundary is zero, equivalently, a $d$-cochain that evaluates to 0 on the boundary of any $(d+1)$-chain. The $d$-cocycles form a subgroup $Z^d(X, A) \leq C^d(X, A)$. A $d$-coboundary is the coboundary of some $(d-1)$-cochain; these form a subgroup $B^d(X, A) \leq Z^d(X, A)$. Two $d$-cochains that differ by a $d$-coboundary are said to be \emph{cohomologous}, and the cohomology classes form a group $H^d(X, A) := Z^d(X, A) / B^d(X, A)$. As with homology, for 2-manifolds the main cohomological interest is in $H^1$.

The support of a $d$-cocycle is the number of $d$-simplices to which it assigns a nonzero value. 

\begin{problem}[1-Cohomology Localization, $\CohLoc(G)$]
Let $G$ be any group.\footnote{If $G$ is Abelian we have already covered the necessarily preliminaries; for the non-Abelian generalization see Appendix~\ref{app:nonab_coho}.} Given a simplicial complex (resp. combinatorial CW complex) $X$ and a 1-cocycle $a \in Z^1(X, G)$, find the sparsest cohomologous representative. 
\end{problem}

On closed surfaces, we have the following equivalence between these problems:


\begin{observation}[J. Erickson, personal communication] \label{obs:erickson}
1-Cohomology Localization on CW complexes that are closed surfaces is equivalent to 1-Homology Localization on CW complexes that are closed surfaces, and dually (swapping the order of homology and cohomology).
\end{observation}

\begin{proof}[Proof sketch]
The proof is essentially Poincar\'{e} Duality. In a bit more detail: Just as one constructs the planar dual graph of a planar graph, given a cell decomposition of a 2-manifold, one constructs the dual cell decomposition. Namely, there is a vertex in the dual for each 2-cell in the primal, there is an edge in the dual crossing each edge in the primal, and a 2-cell in the dual for each vertex in the primal. Think of a 1-chain  as an assignment of elements of $A$ to the edges. Each edge in the primal crosses exactly one edge in the dual; assign the same element of $A$ to the corresponding dual edge. Then one checks that if the primal assignment to edges was a 1-cycle, then the dual assignment is a 1-cocycle in the dual complex, and that homologous 1-cycles are mapped to cohomologous 1-cocycles. The support of the 1-cocycle dual to a 1-cycle is the same as the support of the primal 1-cycle, so the homologous representative of minimum support in the primal is exactly dual to (and has the same support size as) the cohomologous representative of minimum support in the dual. 
\end{proof}

Finally, although we won't really need this until our non-Abelian Appendix~\ref{app:nonab}, the following standard theorem (modulo nomenclature) may be useful for the reader to keep in mind. 
When we consider a $G$-covering space in which the action of $G$ is the action on itself by (left) translations, we have the following equivalence.

\begin{theoremfull} \label{thm:H1cov}
For a group $G$ acting on itself by translations, there is a natural bijection between $H^1(X, G)$ and $G$-covering spaces of $X$ (= principal $G$-bundles over $X$ = regular covering spaces of $X$ with monodromy group contained in $G$).
\end{theoremfull}

Observation~\ref{obs:maxLinCoho} shows that this equivalence extends to the setting of the corresponding optimization problems when $G$ is Abelian, and Observation~\ref{obs:maxLinCoho_nonab} shows the same for general $G$. 

\section{The only known UGC-complete problems are Maximum Section of a \texorpdfstring{$G$-Covering}{G-Covering} Graph} \label{sec:complete}
In this section we carefully write out the proof of Linial's observation \cite{linial} that UG (and $\GammaMaxLin$) is a special case of the following topological problem. We further observe that the reduction between these two problems \cite{KKMO} preserves the topological covering space structure of these problems.

\begin{problem}[Maximum Section of a ($G$-)Covering Graph]
Let $G$ be a group of permutations. Given a graph $X$ and a $G$-covering graph $p\colon Y \to X$, find the partial section of $p$ that is defined on as many edges of $X$ as possible.
\end{problem}

\begin{observation}{Linial \cite[pp.~55--56]{linial}}]\label{obs:UG}
	The Unique Games Problem with labels in $[k]$ is the same as the Maximum Section of a $S_k$-Covering Graph Problem. Furthermore, given two isomorphic $S_k$-covering spaces of the same constraint graph $X$, there is a bijection between assignments to the two instances of UG that exactly preserves the number of constraints satisfied.	
\end{observation}

\begin{proof}
Given an instance of $UG(k)$ on constraint graph $X$, we have seen in Example~\ref{ex:labelExtended} that the label-extended graph $Y$ of this instance is a covering graph of $X$; let $p\colon Y \to X$ be the natural projection. As every constraint is a permutation in $S_k$, $(p,Y)$ is clearly an $S_k$-covering space of $X$. Now suppose $u\colon X \to [k]$ is an assignment to the variables of $X$. We claim that this assignment can be extended to a section of $p$ over the subset of $E(X)$ consisting precisely of those edges of $X$ corresponding to constraints satisfied by $u$. For suppose $(i,j) \in E(X)$ and the constraint on $(i,j)$ is satisfied by $u$, that is, $\pi_{ij}(u(i)) = u(j)$. Then in $Y$, there is an edge from $(i,u(i))$ to $(j,u(j))$ by construction. We may thus extend $u$ to send points of the edge $(i,j) \in E(X)$ to the points of the edge $((i,u(i)), (j,u(j)) \in E(Y)$ bijectively and continuously, and thus extend $u$ to a section of $p$ that is defined over any edge satisfied by $u$.

Conversely, suppose $s\colon X' \to Y$ is a section of the restriction of $p$ to $X' \subseteq X$, that is, $p|_{X'} \colon p^{-1}(X') \to X'$. We may use $s$ to define a partial assignment to the variables of $X$. Namely, for any $i \in V(X')$ (that is, $i \in V(X)$ and $i \in X'$), define $u(i)$ by the equation $s(i) = (i,u(i)) \in V(Y)$. We claim that any edge of $X$ contained entirely in $X'$ is satisfied by this assignment $u$. Indeed, suppose the edge $(i,j) \in E(X)$ is contained entirely in $X'$. As $s$ assigns $u(i)$ to $i$ and $u(j)$ to $j$, and is continuous over all of $X'$, there must be an edge from $s(i) = (i,u(i))$ to $s(j)=(j,u(j))$ in $Y$. But this is the same as saying that $\pi_{ij}(u(i)) = u(j)$, and thus the constraint on this edge is satisfied. Therefore, maximizing the cardinality of the number of edges over which a section of $p$ exists is the same as maximizing the cardinality of the number of constraints satisfied.

Finally, it is a folklore result that there is a bijection between assignments to two instances of $UG(k)$ on the same constraint graph $X$ whose label-extended graphs are isomorphic graph lifts. Indeed, such an isomorphism corresponds simply to re-labeling the domain of each variable. As $S_k$ is the maximal permutation group on a set of size $k$, Observation~\ref{obs:isoGcov} says that two isomorphic $G$-covering spaces of $X$ are isomorphic graph lifts.
\end{proof}




\begin{observation}[{cf. Linial \cite{linial}}] \label{obs:Max2Lin}
Let $A$ be an Abelian group. The $\GammaMaxLin(A)$ Problem is the same as the Maximum Section of an $A$-Covering Graph Problem, where we view $A$ as a permutation group acting on itself by translations. Furthermore, given two isomorphic $A$-covering spaces of the same constraint graph $X$, there is a bijection between assignments to the two instances of UG that exactly preserves the number of constraints satisfied.	
\end{observation}

\begin{proof}[Proof of Observation~\ref{obs:Max2Lin}]
The proof is essentially the same as in Observation~\ref{obs:UG}. However, to show that the label-extended graphs corresponding to instances of $\GammaMaxLin(A)$ have transition functions in $A$, we must note that the constraint $x_i - x_j = c$ ($c \in A$) indeed corresponds to a transition function $\pi_{ij}$ which is given by adding $c$ to the value of $x_i$. If we view the set of vertices in the label-extended graph over a given vertex $i$ as a copy of $A$, then this transition function indeed corresponds to the action of $A$ on itself by translation, so the label-extended graph is an $A$-covering graph.

We note that in this setting, an isomorphism of $A$-covering graphs of the same graph $X$ has the form $x_i \mapsto x_i + c_i$ for each $i \in V(X)$, as in Observation~\ref{obs:isoGcov} / Definition~\ref{def:isoGcov}.
\end{proof}

%
%



Khot, Kindler, Mossel, and O'Donnell \cite{KKMO} prove that $\GammaMaxLin(q)$ is Unique Games-hard by giving a gap-preserving reduction from UG. Our next result is that this reduction sends isomorphic $S_k$-covering spaces to isomorphic $\ZZ{q}$-covering spaces. 

\begin{proposition} \label{prop:KKMO}
	The reduction of \cite{KKMO} from $UG(k)$ to $\GammaMaxLin(q)$ gives a well-defined map from isomorphism classes of $S_k$-covering spaces to isomorphism classes of $\ZZ{q}$-covering spaces.
\end{proposition}

The reduction of \cite{KKMO} is phrased in terms of probabilistic proof checking; here we describe what happens to the underlying constraint graph. Suppose we are given a $UG(k)$ instance with constraint graph $X$ and permutations $\pi_{ij}$ ($(i,j) \in E(X)$). 

\begin{enumerate}
\item (Squaring) First, for every pair of edges incident to the same vertex we create a new edge, where the constraint permutation on the new edge is the product of the permutations on the original pair of edges. We refer to this operation as ``squaring" the instance, as it corresponds to squaring the underlying graph under one of the natural notions of graph product. 

\item (Vector bundling) Next, each vertex is replaced by $q^k$ vertices, identified by $k$-tuples of integers mod $q$, or as KKMO depicted it, coordinates in a $k$-dimensional $q$-ary hypercube. 

\item (Folding) We then identify tuples as equivalent if and only if they differ by a constant vector $(c,c,\dotsc,c) \pmod{q}$. This is a standard process in PCPs known as ``folding.'' Each equivalence class has size $q$, so when we mod out by this relation, each hypercube gets reduced in size by a factor of $q$. If the original graph had $n$ vertices, the new graph will thus have $n q^{k - 1}$ vertices, each of which may be identified uniquely by a pair $(v, [p_1, p_2, \dots, p_k])$, where $v$ is the vertex of the original graph, and $[p_1, p_2, \dots, p_k]$ is the equivalence class of the coordinates under the folding relation. 

\item (Construct constraints) Finally, given a constraint $\pi_{ij} \in S_k$ in the squared graph, we place constraints in the new graph from $(i, [p_1, p_2, \dots, p_k])$ to $(j, [p_1, p_2, \dots, p_k]^{\pi_{ij}})$, where $[\cdot]^{\pi_{ij}}$ denotes the action of $\pi_{ij}$ on $k$-tuples that permutes the coordinates. (Note that since folding mods out by constant vectors $(c,c,\dotsc,c)$, which are invariant under this permutation action, the permutation action is well-defined on the folded coordinates.) The constraint between these two vertices takes the form $x_{(i, [p_1, p_2, \dots, p_k])} - x_{(j, [p_1, p_2, \dots, p_k]^{\pi_{ij}})} = p_1 - p_{\pi_{ij}^{-1}(1)}$. As folding would add $c$ to both $p_1$ and $p_{\pi_{ij}^{-1}(1)}$, the latter value is also well-defined on folded equivalence classes.
\end{enumerate}

There are also non-trivial weights put on the constraints, but the exact values of the weights won't affect our result. For, given two isomorphic $G$-covering spaces over the same weighted graph $X$ (here the weights and the covering space structure do not interact---that is, we have a $G$-covering space over a graph $X$, and we also have weights on the edges of $X$), when we apply the standard reduction to the unweighted case we get two isomorphic $G$-covering spaces over a new unweighted graph $X'$.

It turns out that the KKMO construction is very natural from the topological perspective:

\begin{remark}[For those who know a little about bundles.]
When $q$ is prime, we may treat $\ZZ{q}$ as the field $\F_q$, and then the above construction corresponds to linearizing the $S_k$-bundle by replacing each fiber (of size $k$) by a $k$-dimensional vector space over $\F_q$, with transition maps between them being those induced by the action of $S_k$ on the coordinates of the vector space $\F_q^k$---that is, the natural linearization of the permutation action of $S_k$ on $[k]$. Folding then corresponds to recognizing that the representation of $S_k$ on these vector space fibers is a permutation representation, and therefore constant vectors form an irreducible, trivial sub-representation, and modding out by those. When $q$ is not a prime, essentially the same thing is happening but with $\ZZ{q}$-module bundles rather than vector bundles.
\end{remark}

\begin{proof}[Proof of Proposition~\ref{prop:KKMO}]
	First we show that the squaring map is a well-defined map from isomorphism classes of $S_k$-covering spaces to isomorphism classes of $S_k$-covering spaces. Given an $S_k$-covering space $p\colon Y \to X$, let $X^2$ denote the ``square'' of $X$ (which has an edge $(i,j)_k$ for each pair of edges $(i,k), (k,j) \in E(X)$), and let $Y^2$ denote the corresponding $S_k$-covering space, namely where the transition function $\pi_{(i,j)_k} = \pi_{ik} \pi_{kj}$. Then we claim that given two $S_k$-covering spaces $p_i\colon Y_i \to X$ ($i=1,2$), if they are isomorphic as covering spaces, then $Y_1^2$ and $Y_2^2$ are isomorphic as covering spaces of $X^2$. 
	
	Let $\pi_{ij}^{(\ell)}$ denote the transition function along edge $(i,j) \in E(X)$ in the covering space $Y_\ell$. As $Y_1, Y_2$ are isomorphic $S_k$-covering spaces, by Observation~\ref{obs:isoGcov} there are permutations $\pi_i$ ($i \in V(X)$) such that $\pi_{ij}^{(2)} = \pi_i^{-1} \pi_{ij}^{(1)} \pi_j$ for all $(i,j) \in E(X)$. We claim that these same $\pi_i$ define an isomorphism of $Y_1^2$ and $Y_2^2$:
	\begin{eqnarray*}
	\pi_{(i,j)_k}^{(2)} & = & \pi_{ik}^{(2)} \pi_{kj}^{(2)} \\
	& = & \pi_i^{-1} \pi_{ik}^{(1)} \pi_k^{-1} \pi_k \pi_{kj}^{(1)} \pi_j \\
	& = & \pi_i^{-1} \pi_{ik}^{(1)} \pi_{kj}^{(1)} \pi_j \\
	& = & \pi_i^{-1} \pi_{(i,j)_k}^{(1)} \pi_j.
	\end{eqnarray*}
Thus, using the second half of Observation~\ref{obs:isoGcov}, we see that $Y_1^2$ and $Y_2^2$ are isomorphic $S_k$-covering spaces.
	
	Next, we show that the remainder of the construction respects isomorphism classes of $G$-covering spaces, that is, given two isomorphic $S_k$-covering spaces, the output of steps (2)--(4) are two isomorphic $\ZZ{q}$-covering spaces. Continuing with the notation above, suppose $Y_1^2, Y_2^2$ are the label-extended graphs of the squared $UG(k)$ instances, and they are isomorphic $S_k$-covering spaces via the permutations $\pi_i$ ($i \in V(X)$). 
	
	For $\ell=1,2$, let $W_\ell$ be the constraint graph of the instance of $\GammaMaxLin(q)$ output by the KKMO construction, and let $Z_\ell$ be the label-extended graph of $W_\ell$ (that is, a $\ZZ{q}$-covering space of $W_\ell$). One issue with comparing $Z_1$ and $Z_2$ is that $W_1$ and $W_2$ need not be the same graph, even though they share the same underlying vertex set. Nonetheless, we will show that they are isomorphic graphs, and that under a suitable isomorphism $W_1 \to W_2$---induced by the isomorphism of $S_k$-covering spaces $Y_1 \to Y_2$---the $\ZZ{q}$ covering spaces $Z_1, Z_2$ then also become isomorphic. 
	
	First we define the isomorphism $\varphi\colon W_1 \to W_2$. As above, let $\pi_i$ ($i \in V(X)$) be the permutations realizing the isomorphism $Y_1 \to Y_2$ (hence, also the isomorphism $Y_1^2 \to Y_2^2$). Then we define $\varphi$ on vertices by
	\[
	\varphi((i,[p_1,p_2,\dotsc,p_k])) = (i,[p_1,\dotsc,p_k]^{\pi_i}).
	\]
	That is, $\varphi$ does not change the underlying vertex $i$ of $X$, it only moves around the vertices of $W_*$ lying over a given vertex of $X$, by scrambling the hypercube coordinates by $\pi_i$. To see that this is an isomorphism, note that $w=(i,[p_1,\dotsc,p_k])$ is adjacent in $W_1$ to $w' = (j, [p_1', \dotsc, p_k'])$ if and only if $(i,j) \in E(X)$ and $[p_1,\dotsc,p_k]^{\pi_{ij}^{(1)}} = [p_1',\dotsc,p_k']$, by construction. We may combine this with the fact that $\pi_{ij}^{(2)} = \pi_i^{-1} \pi_{ij}^{(1)} \pi_j$, to get: 
	\begin{eqnarray*}
	(w,w') \in E(W_1) & \Leftrightarrow & ((i,[p_1,\dotsc,p_k]), (j,[p_1', \dotsc, p_k'])) \in E(W_1) \\
	& \Leftrightarrow & (i,j) \in E(X) \text{ and } [p_1,\dotsc,p_k]^{\pi_{ij}^{(1)}} = [p_1',\dotsc,p_k'] \\
	& \Leftrightarrow & (i,j) \in E(X) \text{ and } [p_1,\dotsc,p_k]^{\pi_i \pi_{ij}^{(2)} \pi_j^{-1}} = [p_1',\dotsc,p_k'] \\
	& \Leftrightarrow & (i,j) \in E(X) \text{ and } \left([p_1,\dotsc,p_k]^{\pi_i}\right)^{\pi_{ij}^{(2)}} = [p_1',\dotsc,p_k']^{\pi_j} \\
	& \Leftrightarrow & ((i, [p_1,\dotsc,p_k]^{\pi_i}), (j,[p_1',\dotsc,p_k']^{\pi_j})) \in E(W_2) \\
	& \Leftrightarrow & (\varphi(w), \varphi(w')) \in E(W_2).
	\end{eqnarray*}
	Thus $\varphi\colon W_1 \to W_2$ is indeed an isomorphism. This allows us to treat $Z_2$ as a $\ZZ{q}$-covering space of $W_1$ (sic), by considering the following composition: $Z_2 \to W_2 \stackrel{\varphi^{-1}}{\to} W_1$. 
	
	Finally, we will show that $Z_1$ and $Z_2$ are isomorphic $\ZZ{q}$-covering spaces of $W_1$. To show this, let us first see what element of $\ZZ{q}$ each $Z_\ell$ assigns to each edge of $W_1$. Given an edge from $(i,[p_1,\dotsc,p_k])$ to $(j, [p_1,\dotsc,p_k]^{\pi_{ij}^{(1)}})$ in $W_1$, in $Z_1$ it is assigned the group element \[
	p_1 - p_{(\pi_{ij}^{(1)})^{-1}(1)}.
	\]
	 To see what this is assigned to by $Z_2$, first let us use $\varphi$ to turn this into an edge of $W_2$: it becomes the edge from $(i,[p_1,\dotsc,p_k]^{\pi_i})$ to $(j, [p_1,\dotsc,p_k]^{\pi_{ij}^{(1)}\pi_j})$. $Z_2$ thus assigns this edge the group element 
	 \[
	 p_{\pi_i^{-1}(1)} - p_{\pi_j^{-1} (\pi_{ij}^{(1)})^{-1}(1)}.
	 \]
If we change the variable at the vertex $(i,[p_1,\dotsc,p_k])$ by the $\ZZ{q}$ translation $c(i,[p_1,\dotsc,p_k])$ defined by:
	\[
	x_{i,[p_1,\dotsc,p_k]} \mapsto x_{i,[p_1,\dotsc,p_k]} + p_{\pi_i^{-1}(1)} - p_1,
	\]
	we will see this gives the desired isomorphism of $\ZZ{q}$-covering spaces. Consider how this changes the group element assigned to the edge from $(i,[p_1,\dotsc,p_k])$ to $(j,[p_1,\dotsc,p_k]^{\pi_{ij}^{(1)}})$. It changes from $p_1 - p_{(\pi_{ij}^{(1)})^{-1}(1)}$ to:
	\begin{eqnarray*}
	& = & p_1 - p_{(\pi_{ij}^{(1)})^{-1}(1)}  + c(i,[p_1,\dotsc,p_k]) - c(j,[p_1,\dotsc,p_k]^{\pi_{ij}^{(1)}}) \\
	& = & p_1 - p_{(\pi_{ij}^{(1)})^{-1}(1)} + (p_{\pi_i^{-1}(1)} - p_1) - (p_{\pi_j^{-1}((\pi_{ij}^{(1)})^{-1}(1))} - p_{(\pi_{ij}^{(1)})^{-1}(1)}) \\
	& = & p_{\pi_i^{-1}(1)} - p_{\pi_j^{-1} (\pi_{ij}^{(1)})^{-1}(1)}
	\end{eqnarray*}
	
	This concludes the proof.
\end{proof}

\begin{remark} \label{rmk:principal}
It is well-understood that one difference between $UG(k)$ and $\GammaMaxLin(k)$ is that if an instance of $\GammaMaxLin(k)$ is satisfiable, then one can choose an arbitrary vertex, assign an arbitrary value in $\ZZ{k}$ to this vertex, and propagate this value across the entire graph using the constraints, and this will always be a solution. In contrast, even if an instance of $UG(k)$ is satisfiable, starting from a given vertex there may be (in the worst case) only one value that can be assigned to that vertex in such a way that the propagated assignment is actually satisfying. 

In topological language, the above translates nearly exactly as follows: A $G$-covering space for general $G$ is a $G$-bundle with finite fibers. When the action of $G$ is the action on itself by translations, this yields what's called a \emph{principal} $G$-bundle. When $A$ is Abelian and the action is faithful, every $A$-covering space is a principal $A$-bundle, whereas this need not be the case for non-Abelian groups. The translation of the preceding paragraph into this language is the following property of principal bundles: A principal bundle is trivial (a direct product) if and only if it has a section and if so, such a section can be found by making an arbitrary choice at one vertex and propagating.
\end{remark}

\section{1-Homology Localization on cell decompositions of 2-manifolds is UGC-complete} \label{sec:manifolds}

\begin{theorem} \label{thm:homLoc}
1-Homology Localization on cell decompositions of closed orientable surfaces is Unique Games-complete.

More precisely, the Unique Games Conjecture holds if and only if for any $\varepsilon,\delta > 0$, there is some $k = k(\varepsilon,\delta)$ such that $\Gap\HomLoc_{1-\varepsilon,\delta}$ on cell decompositions of closed orientable surfaces with coefficients in $\ZZ{k}$ is $\cc{NP}$-hard.
\end{theorem}

To make the equivalence here a bit cleaner, we introduce a small variant of $\GammaMaxLin$ on surfaces instead of graphs:

\begin{problem}[$\GammaMaxLin(A)$ on surfaces]
Let $A$ be an Abelian group. Given a cell decomposition of a closed surface $X$, and a 1-cocycle on $X$ with coefficients in $A$ treat the 1-cocycle as defining an instance of $\GammaMaxLin(A)$. In other words, this problem is the same as $\GammaMaxLin(A)$ on the 1-skeleton $X_1$ of $X$, except that we only consider instances of $\GammaMaxLin(A)$ in which the sum of the constraints along each cycle of $X_1$ that is the boundary of a 2-cell of $X$ is zero.
\end{problem}

From the second characterization in the definition (``in other words...''), it is clear that $\GammaMaxLin(A)$ on cell decompositions of surfaces is potentially easier than $\GammaMaxLin(A)$ on graph. We show that $\GammaMaxLin(A)$ on surfaces nonetheless remains UGC-complete.

\begin{proof}[Proof of Theorem~\ref{thm:homLoc}]
The proof proceeds as follows, and will take up the remainder of this section:
\[
\GammaMaxLin(A) \text{ on graphs} \leq \GammaMaxLin(A) \text{ on surfaces} \cong \CohLoc \text{ on surfaces} \cong \HomLoc \text{ on surfaces},
\]
where all reductions here are gap-preserving reductions.
The first reduction is the technically tricky part, embodied in Proposition~\ref{prop:graphToSurf} below. The second equivalence is Observation~\ref{obs:maxLinCoho}, which we do first since it will inform some of our subsequent discussion. The final equivalence is Observation~\ref{obs:erickson}.
\end{proof}

\begin{observation} \label{obs:maxLinCoho}
$\GammaMaxLin(A)$ on a cell decomposition of a surface $X$ is equivalent (under gap-preserving reductions) to $\CohLoc(A)$ on the same cell decomposition of the same surface.
\end{observation}

\begin{proof}
Given an instance of $\GammaMaxLin(A)$ specified by constants $a_{ij} \in A$ (that is, with constraints $x_i - x_j = a_{ij}$ for all edges in the cell decomposition of $X$), by the definition of $\GammaMaxLin(A)$ on surfaces the function $a\colon (i,j) \mapsto a_{ij}$ is a 1-cocycle. 
We claim that the following is a natural bijection between assignments to the variables and cohomologous 1-cocycles such that the set of constraints satisfied by an assignment is precisely the complement of the support of the associated 1-cocycle: Given an assignment $\alpha_i$ to the variables $x_i$, treat $\alpha$ as a 0-cochain, and consider the 1-cocycle $a - \delta \alpha$. Note that $\alpha$ satisfies some edge $(i,j)$ iff $(a - \delta \alpha)(i,j) = 0$, for $(a - \delta \alpha)(i,j) = a_{ij} - \alpha_i + \alpha_j$. 

Conversely, given a cohomologous 1-cocycle $a'$, the difference $a - a'$ is the coboundary of some 0-cochain: $a - a' = \delta \alpha$; treat $\alpha$ as an assignment to the variables. Using the same equation as before, we see that the support of $a'$ is precisely the complement of the set of constraints satisfied by $\alpha$.

Thus maximizing the number of satisfied constraints is equivalent to minimizing the support of a cohomologous cocycle. All that remains to check is that the equivalence above is indeed gap-preserving, noting that $\GammaMaxLin(A)$ is a maximization problem while $\CohLoc(A)$ is a minimization problem. Here we take the value of a cocycle to be the \emph{fraction} of nonzero edges. The above shows that if the maximum fraction of satisfiable constraints in the $\GammaMaxLin(A)$ instance is $\rho$, then the minimum fraction of edges in the support of a cohomologous cocycle is $1-\rho$ (since the number of edges in the same in both instances). So if a $\geq 1-\varepsilon$ fraction of the constraints are satisfiable, then there is a cohomologous cocycle with support consisting of $\leq \varepsilon$ fraction of the edges; and if a $\leq \delta$ fraction of the constraints are satisfiable, then every cohomologous cocycle contains a $\geq 1-\delta$ fraction of the edges.
\end{proof}

Our strategy for the reduction from $\GammaMaxLin(A)$ on graphs to $\GammaMaxLin(A)$ on surfaces will be to take an arbitrary graph and embed it as the 1-skeleton of a closed surface in polynomial time, in a gap-preserving way. 
The complication is that we must be careful about adding 2-cells. 
Given an instance of $\GammaMaxLin(A)$ on a graph $X$, and some cycle in $X$ such that the sum of the constraints around the cycle is nonzero, we cannot simply ``fill in" the cycle with a 2-cell. If we added such a cell to the complex, then the instance would no longer correspond to a 1-cocycle. Thus, we may only add 2-cells to cycles that are satisfiable in the given instance.

Furst, Gross, and McGeoch \cite{furstGrossMcGeoch} 
give a polynomial-time algorithm to find the \emph{maximal genus embedding} of a graph, that is, an embedding on an orientable closed surface such that the complement of the graph decomposes into a disjoint union of disks, in such a way as to maximize the genus of the surface. 
In a cellular embedding, Euler's polyhedral formula applies:
$V - E + F = 2 - 2g,$
where $V$, $E$, and $F$ are the numbers of vertices, edges, and faces of the complex, and $g$ is the genus of the surface. Holding $V$ and $E$ fixed, we see that the problem of maximizing the genus is equivalent to minimizing the number of faces---and hence minimizing the additional properties an instance of $\GammaMaxLin(A)$ on a surface must satisfy compared to being on a graph. In the extreme case, their algorithm may find an embedding with only one face, which ``wraps around" the surface touching every edge twice, once on each side. Because the orientations of the region on either side oppose each other, the boundary of this region will always be zero. Therefore, if such a region is added to the complex, \emph{any 1-cochain will still be a 1-cocycle}.

The algorithm of Furst, Gross, and McGeoch \cite{furstGrossMcGeoch} is based on the work of Xuong \cite{xuong}, who gave a complete characterization of how many regions one needs to embed a graph. We restate the special case of his main result in which only one region is needed.

\begin{theorem}[Xuong \cite{xuong}] \label{thm:xuong}
	A connected graph $G$ has a one-face cellular embedding into a closed orientable surface if and only if there exists a spanning tree $T$ such that every connected component of $G \setminus T$ has an even number of edges.
\end{theorem}

Since this wasn't explicitly stated in Xuong \cite{xuong}, and an anonymous reviewer had a question about it, we show how it easily follows from the results as they are stated in \cite{xuong}. Lemmas 3.1 \& 3.3 of Furst, Gross, and McGeoch \cite{furstGrossMcGeoch} together show the same result, with the spanning tree satisfying an additional condition.

\begin{proof}
Xuong's main result is that the maximum genus $\gamma_M$ of $G$ is equal to $\frac{1}{2}(\beta(G) - \xi(G))$, where $\beta = E - V + 1$ is the Betti number of $G$ and 
\[
\xi(G) = \min \{ \# \text{ components of } G \backslash T \text{ with an odd number of edges} : T \text{ a spanning tree of } G\}.
\]

Suppose $G$ has a one-face cellular embedding into a closed orientable surface of genus $g$. Then, by Euler's formula, $1 - E + V = 2 - 2g$. Fixing $E,V$, the number of faces decreases iff the genus increases, so a 1-face cellular embedding must also be a max-genus embedding. Therefore $1 - E + V = 2 - 2\gamma_M(G)$, or equivalently $\gamma_M(G) = \frac{1}{2}(1 + E - V) = \frac{1}{2}\beta(G)$. Then Xuong's theorem implies $\xi(G)=0$, that is, there is a spanning tree $T$ such that every component of $G \backslash T$ has an even number of edges.

Conversely, suppose there is a spanning tree $T$ such that every component of $G \backslash T$ has an even number of edges. Then $\xi(G)=0$ and Xuong's theorem gives $\gamma_M(G) = \frac{1}{2}\beta(G)$. Using the calculation with Euler's formula from the preceding paragraph, we find that any cellular embedding of $G$ into a closed orientable surface of genus $\gamma_M(G)$ must thus have exactly one face. By Xuong's thereom, such an embedding exists.
\end{proof}

We are now ready to formally describe our reduction and prove its correctness.

\begin{proposition} \label{prop:graphToSurf}
There is a reduction of the form in Lemma~\ref{LinEdges} from $\GammaMaxLin(A)$ on graphs to $\GammaMaxLin(A)$ on cell decompositions of surfaces, and therefore $\GammaMaxLin(A)$ on surfaces is UGC-complete.
\end{proposition}

\begin{proof}
	Suppose we are given a 1-cocycle on a graph $X$ (every 1-cochain on a graph is a 1-cocycle, since there are no 2-cells). We describe an algorithm to transform $X$ into a graph $X'$ that admits a cellular embedding with one region. First, as $X$ is the constraint graph of a $\GammaMaxLin(A)$ instance, we may assume without loss of generality that $X$ has no isolated vertices. Now, if $X$ is disconnected, connect the components together, using one edge for each extra component.  If the total number of edges is now odd, append a leaf attached by an edge to any vertex. Finally, add a ``universal'' vertex $u$, with an edge to every other vertex. Each new edge added in the preceding steps can be labeled with any constraint; it will not matter, but say $x_i - x_j = 0$ for concreteness. Let $X'$ be the resulting graph. Letting $T$ be the star spanning tree centered at $u$---that is, consisting of precisely the edges incident on $u$---we see that $X' \setminus T$ has only one component, with an even number of edges, so $X'$ can be embedded in polynomial time with one region. The fact that this reduction preserves the inaproximability gap is the content of Lemma \ref{LinEdges}: If $\GammaMaxLin(A)$ was hard on arbitrary graphs, it will still be hard on graphs with one-face embeddings, as this reduction added at most $(\frac{v}{2} - 1) + 1 + (v + 1) = \frac{3}{2}v + 1 = O(v)$ edges. 
\end{proof}

This completes the proof of Theorem~\ref{thm:homLoc}.

We would like to further reduce the instance so that it is a simplicial complex, rather than just a cell decomposition, which would prove that the simplicial version of the problem is UGC-hard as well. However, to break a 1-region embedding into triangles seems to require so many edges that it seems impossible to preserve the inapproximability gap completely. In the next section, we manage to preserve a $7/6$ gap.


With $G$-covering spaces suitably defined, we show essentially the same result for non-Abelian $G$. We use $G$-covering spaces rather than \CohLoc because \CohLoc only corresponds to principal $G$-covering spaces, in which $G$ acts on itself by translations. Using the following theorem, we could have shown UGC-completeness of Maximum Section of a $G$-Covering Space on cell decompositions of 2-manifolds without going through $\MaxLin$; we chose the above route as the concepts with Abelian coefficient groups are simpler and more well-known. For the needed non-Abelian definitions and for the proof of this next result, see Appendix~\ref{app:nonab}.

\begin{theorem} \label{thm:nonab}
Let $G$ be any group such that every product of commutators of $G$ can be written as a singular commutator. Then Maximum Section of $G$-Covering Spaces on graphs reduces to Maximum Section of $G$-Covering Spaces on cell decompositions of 2-manifolds. In particular, the latter problem for ($S_k$-)covering spaces of cell decompositions of 2-manifolds is UGC-complete.
\end{theorem}

Ore \cite{ore} showed that every element of the alternating group $A_k$ is a single commutator (giving the ``in particular'' of the above theorem). More generally, Ore conjectured \cite{ore} that all finite simple groups had the property that every element was a commutator. This conjecture was proved by Liebeck, O'Brien, Shalev, and Huu Tiep \cite{LOST}, which they then extended to all quasi-simple groups except for a small explicit list of exceptions \cite{LOST2} (a group is quasi-simple if $G=[G,G]$ and $G/Z(G)$ is simple; these include groups such as $\SL_n(\F_q)$). For more on commutators, Ore's Conjecture, and its generalizations, see Kappe and Morse \cite{KM}; although their paper predates \cite{LOST,LOST2}, it still gives valuable perspective. Although this already gives a significant number of groups satisfying the hypothesis of Theorem~\ref{thm:nonab}, note that we don't actually need groups where \emph{every} element is a commutator for Theorem~\ref{thm:nonab}, we only need groups $G$ such that every element of $[G,G]$ is a commutator in $G$.

\section{On a question of Chen and Freedman} \label{sec:CF}
\begin{quotation}
\noindent``...we prove [Homology Localization over $\ZZ{2}$] is $\cc{NP}$-hard to approximate for one-dimensional homology when the input is a triangularization of a 3-manifold. 
This raises the open question [of] whether localizing a one-dimensional class of a 2-manifold is $\cc{NP}$-hard to approximate...'' --Chen and Freedman \cite[p.~438]{chenFreedman}
\end{quotation}

Note that Chen and Freedman showed that $d$-Homology Localization, for any fixed $d \geq 2$ is indeed $\cc{NP}$-hard to approximate to within any constant factor for triangulations of manifolds (of unbounded dimension), as well as \HomLoc for triangulations of 3-manifolds. One might thus infer from the above quote that they were asking the same question---inapproximability to within any constant factor---for triangulations of 2-manifolds. We note that although a greedy algorithm gives a $k$-approximation to Unique Games over $\ZZ{k}$ \cite[Appendix]{SG}, when we translate this maximization problem to the minimization problem of \HomLoc, we have essentially no control over the approximation ratio. We nonetheless show some inapproximability, by a different method.

\begin{theorem} \label{thm:76}
Assuming UGC, for any $\varepsilon > 0$, there is some $k=k(\varepsilon)$ such that it is hard to approximate \HomLoc over $\Z/k\Z$ on triangulations of surfaces to within a factor of $7/6-\varepsilon$.
\end{theorem}

\begin{proof}
Let $\varepsilon_0, \delta_0 > 0$ be such that $\frac{7-5\varepsilon_0}{6 + 3\delta_0} > 7/6 - \varepsilon$. Let $k$ be sufficiently large so that $\Gap\GammaMaxLin(k)_{1-\varepsilon_0,\delta_0}$ is hard assuming UGC \cite{KKMO}. 
Given an instance of $\GammaMaxLin(k)$ on a graph, first we reduce it to an instance of $\GammaMaxLin(k)$ with the same satisfiable ratio (up to an additive $o(1)$) using Proposition~\ref{prop:graphToSurf}. Now we have a cell decomposition of a 2-manifold with just a single 2-cell. We will triangulate this 2-cell in a way that gives us the desired $7/6$ approximation ratio. 

Let $X'$ be the graph, with $v$ vertices and $e$ edges, which is embedded with one region. We will construct a graph $X''$ containing $X'$ as follows. The reader may find it useful to consult Figure~\ref{fig:simplicial} along with this description. Start with the vertex set of $X'$, and add a new vertex $v_i$ just inside of each edge $i$ around the region. Note that we are counting both times the edge appears, so $i$ will range from 0 to $2e - 1$, which we will label increasing in the clockwise direction. Add one more vertex $u$ in the very center of the region. We define three graphs with this vertex set. Let $S$ contain the star of edges from $u$ to $\{v_i : 0 \leq i < 2e\}$. Let $C$ contain the cycle of edges from $v_i$ to $v_{(i + 1 \bmod 2e)}$ for each $0 \leq i < 2e$. Let $L$ contain edges $e_{i,\ell}$ and $e_{i,r}$ adjoining each $v_i$ to both endpoints of the edge $i$, going to the left and right, respectively. We set the constraints on each $e_{i, \ell}$ to be zero, and arbitrarily choose an edge in $S$ to have a zero constraint as well. The constraint on every other edge can be deduced as the unique constraint that is satisfied if the other two edges in the triangle are satisfied. We define $X'' := X' \cup S \cup C \cup L$. Note that if we break the region into the triangles defined by $X''$, we will be left with a simplicial complex.

\begin{figure}[!htbp]
\begin{center}
\includegraphics[scale = .4]{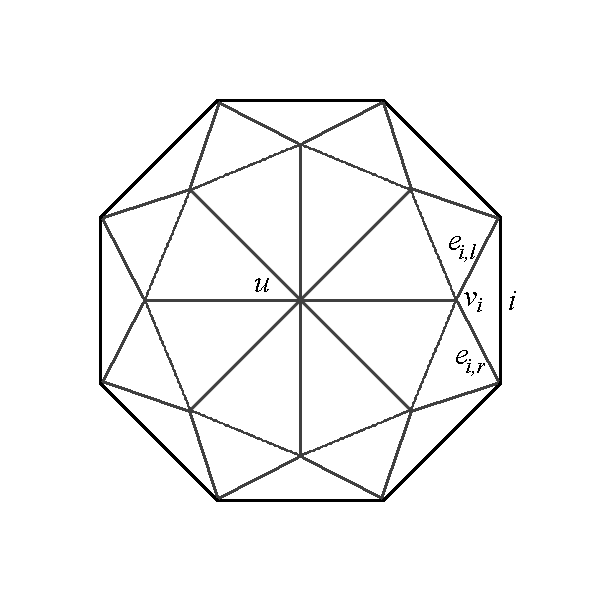}
\end{center}
\caption{Example of the triangulation construction.}
\label{fig:simplicial}
\end{figure}

It is clear by construction that the edge set of $X''$ is a disjoint union on the edge sets of $X'$, $S$, $C$, and $L$, and that the proportion of edges in each subgraph is as follows:
$$\frac{|E(X'')|}{1} = \frac{|E(S)|}{2} = \frac{|E(C)|}{2} = \frac{|E(L)|}{4} = \frac{|E(H)|}{9}.$$
Thus, to compute the satisfiability of $X''$, we can take a weighted average over the satisfiability of the subgraphs. Letting $s$ denote the fraction of satisfied edges of a subgraph in an optimal assignment to the \emph{entire} graph, we have
$$s(X'') = \frac{s(X') + 2s(S) + 2s(C) + 4s(L)}{9}.$$

In the following lemmas, we will use this formula to argue that this construction yields a $7/6$-inapproximability gap: 
We show that if $s(X') \leq \delta_0$, then $s(X'') \leq \frac{2 + \delta_0}{3}$ and if $s(X') \geq 1-\varepsilon_0$ then $s(X'') \geq \frac{7-5\varepsilon_0}{9}$. Thus, assuming UGC it is hard to approximate \HomLoc over $\Z/k\Z$ for triangulations of 2-manifolds to within $\frac{7-5\varepsilon_0}{9} \times \frac{3}{2 + \delta_0} > 7/6 - \varepsilon$, and the proofs of the lemmas will complete the proof of the theorem.
\end{proof}

\begin{lemmafull}
With notation as above, if $s(X') \leq \delta_0$, then $s(X'') \leq \frac{2 + \delta_0}{3}$.
\end{lemmafull}

\begin{proof}
	Consider the edges in $L$. For each side of each edge $i$ that is not satisfied, at least one of $e_{i,\ell}$ or $e_{i,r}$ must also be unsatisfied, for otherwise we would have only two out of three edges satisfied in a triangle, which is impossible by construction. At least a $1 - \delta_0$ fraction of these edges in $X'$ must be unsatisfied, so therefore, at least a $\frac{1 - \delta_0}{2}$ fraction in $L$ must be unsatisfied as well. Thus, $s(L) \leq 1 - \frac{1 - \delta_0}{2} = \frac{1 + \delta_0}{2}$, and so
	$$s(X'') = \frac{s(X') + 2s(S) + 2s(C) + 4s(L)}{9} \leq \frac{\delta_0 + 2 + 2 + 4\frac{1 + \delta_0}{2}}{9} = \frac{6 + 3\delta_0}{9} = \frac{2 + \delta_0}{3}.$$
\end{proof}

\begin{lemmafull}
With notation as above, if $s(X') \geq 1 - \varepsilon_0$, then $s(X'') \geq \frac{7 - 5\varepsilon_0}{9}$.
\end{lemmafull}

\begin{proof}
	We describe a labeling that will be guaranteed to satisfy a $\frac{7 - 5\varepsilon_0}{9}$ fraction of edges in $X''$. First pick some labeling of $X'$ that satisfies at least $1 - \varepsilon_0$ of the edges. Assign each $v_i$ the same label as the left endpoint of edge $i$, so that the constraint on $e_{i,\ell}$ is satisfied. Choose any label for the vertex $u$. We already know that half of the edges in $L$, namely the $e_{i,\ell}$'s, will be automatically satisfied. For the other half, note that if an edge $i$ in $X'$ is satisfied, then $e_{i,r}$ must be satisfied as well, because the constraint was chosen so that it would be satisfied if both $e_{i,\ell}$ and $i$ were satisfied. Because this happens for at least a $1 - \varepsilon_0$ fraction of the edges in $X'$, we conclude that at least a $1 - \varepsilon_0$ fraction of the $e_{i,r}$ edges are satisfied, for a total of $s(L) = \frac{1 + (1 - \varepsilon_0)}{2} = \frac{2 - \varepsilon}{2}$ edges overall. Now consider the edges in $C$. For each $0 \leq i < 2e$, letting $j = i + 1 \bmod 2e$, the edge from $v_i$ to $v_j$ in $C$, which we will call $c$, forms a triangle with $e_{i,r}$ and $e_{j,\ell}$. We know that $e_{j,\ell}$ is satisfied, which means that $c$ is satisfied if and only if the third edge, $e_{i,r}$, is satisfied. But we know that at least a $1 - \varepsilon_0$ fraction of the $e_{i,r}$ edges are satisfied, so at least a $1 - \varepsilon_0$ fraction of the edges in $C$ must be satisfied as well. Putting together these bounds on the satisfiability of $X'$, $C$, and $L$, we conclude that
	$$s(X'') = \frac{s(X') + 2s(S) + 2s(C) + 4s(L)}{9} \geq \frac{(1 - \varepsilon_0) + 2(0) + 2(1 - \varepsilon_0) + 4\left(\frac{2 - \varepsilon_0}{2}\right)}{9} = \frac{7 - 5\varepsilon_0}{9}.$$
\end{proof}

Elsewhere in their paper they consider cell decompositions. If we relax their question slightly to cell decompositions rather than triangulations, and we allow the coefficient group to be $\ZZ{k}$ for arbitrary $k$ (rather than only $k=2$), then Theorem~\ref{thm:homLoc} states that their question becomes equivalent to UGC.

\section{Future directions} \label{sec:future}
Although our most technically demanding results were applying UGC to 1-Homology Localization---and, in the course of this, showing a new UGC-complete problem---we hope that the connections we have drawn between UGC and computational topology will lead to further progress on both topics in the future. Here we highlight a few specific questions suggested by our investigations.

\paragraph{Inapproximability of 1-Homology Localization for triangulations of 2-manifolds?}
Although our results partially settle a question of Chen and Freedman \cite[p.~438]{chenFreedman}, we leave open the following questions:

\begin{question}
Show (unconditionally) that there is a $c > 1$ such that it is $\cc{NP}$-hard to $c$-approximate 1-Homology Localization over $\ZZ{2}$ on triangulations of 2-manifolds.
\end{question}

\begin{question}
Does UGC imply that for all $c > 1$, it is hard to $c$-approximate 1-Homology Localization on triangulations of 2-manifolds (over $\ZZ{k}$ for $k=k(c)$)?
\end{question}

We note that our reduction from Theorem~\ref{thm:76} does not provide a strong enough gap for these problems to be immediately answered by the known $\cc{NP}$-hardness results for $\MaxLin$ \cite{hastad, AEH}. In particular, H\aa stad shows that $\Gap\MaxLin(2)_{\frac{12}{16},\frac{11}{16}}$ is $\cc{NP}$-hard (up to an additive arbitrary $\delta$ in the soundness and completeness), but if we plug these values into $\varepsilon_0, \delta_0$ from our proof, we get a ratio of $\frac{92}{129} < 1$. A quick calculation shows that, to get any $\cc{NP}$-hardness result (without UGC) from the proof of Theorem~\ref{thm:76}, we would need an inapproximability ratio for $\MaxLin$ of strictly greater than $\frac{12}{5}$, which is not provided by \cite{hastad} nor \cite{AEH}. 
Given the best upper bounds on approximating $\MaxLin(p)$ \cite{AEH}, which are just slightly less than $p$, such a ratio is not possible for $p=2,3$, but is possible already for $p=5$.

\paragraph{$G$-covering spaces for other families of groups and group actions.}
The viewpoint of $G$-covering spaces suggests it might be fruitful to consider instances of UG that correspond to $S_n$-covering spaces for other actions of $S_n$. For example, one might consider the action of $S_n$ on unordered $k$-tuples $\binom{[n]}{k}$, or even on $n$-vertex graphs $2^{\binom{[n]}{2}}$ (here, each variable in the UG would have the set of $n$-vertex graphs as its domain). We note that although much of the structure of any such instance is governed by the permutation constraints, the approximability properties may change significantly by varying the action. The general linear groups $\GL_n(\F_q)$, acting on the vector space $\F_q^n$, as well as Schur functors thereof, or other representations of $\GL_n(\F_q)$, strike us as leading to other possibly interesting approximation problems deserving further study. If one is looking for hard instances of Unique Games, one must construct such covering spaces so that they are not (sufficiently good) expanders \cite{AKKSTV}; see \cite{ACKM} for results on the expansion of $G$-covering spaces.

\section*{Acknowledgments}

We thank Alex Kolla for a great talk on Unique Games that inspired us in this direction, we thank Ryan O'Donnell and Hsien-Chih Chang for useful discussions, and we thank Jeff Erickson for answering some of our questions and providing references to the computational topology literature. We especially thank Thomas Church for helping us clarify some of the topological issues involved and suggesting how to modify the proof of Theorem~\ref{thm:homLoc} to extend it to the $S_n$ case (Theorem~\ref{thm:nonab}).
J. A. G. gratefully acknowledges support from NSF grants DMS-1750319 and DMS-1622390 during the course of this work. J. T.-F. gratefully acknowledges the support of his Schupf Scholarship at Amherst College.

\appendix

\section{Covering spaces and cohomology with non-Abelian coefficient groups \texorpdfstring{$G$}{G}} \label{app:nonab}

In this appendix we cover generalizations of the main notions in the paper (a) from graphs to arbitrary topological spaces (though we will focus on 2-dimensional combinatorial CW complexes) and (b) from Abelian to non-Abelian groups. This will allow us to make precise and prove Theorem~\ref{thm:nonab} (in Section~\ref{app:nonab_proof}), as well as to make precise and prove the equivalence of \CohLoc with coefficients in an arbitrary group $G$ and Maximum Section of Principal $G$-Covering Spaces (Observation~\ref{obs:maxLinCoho_nonab}).

\subsection{An introduction to \texorpdfstring{$G$-covering}{G-covering} spaces over arbitrary topological spaces}
\emph{Terminological note:} What we refer to here as ``$G$-covering spaces'' might be more typically referred to as ``covering spaces which are also $G$-bundles'' or ``$G$-bundles with finite fibers''; we have chosen the more compact nomenclature for brevity, because covering spaces may be more familiar to more of our readers, and to avoid discussing fiber bundles as much as possible. For more information, we refer the reader to Hatcher \cite{hatcher} for general algebraic topology and covering spaces, and to Husem\"{o}ller \cite{husemoller} for $G$-bundles.

The aim of this section is to recall the generalization of $G$-covering graphs to arbitrary topological spaces. 
In an effort to make this paper readable by as broad an audience as possible, we cover some standard facts from algebraic topology; we hope readers familiar with these facts will take no offense.


A \emph{map} between two topological spaces is a continuous function.

\begin{definition}[Covering space]
Given a topological space $X$ (graph, 2-manifold, simplicial complex, etc.), a \emph{($k$-sheeted) covering space} of $X$ is a topological space $Y$ together with a map $p\colon Y \to X$ such that, for each $x \in X$, there is an open neighborhood $U$ of $x$ such that $p^{-1}(U)$ is a disjoint union of open sets $U_1, \dotsc, U_k \subseteq Y$ each of which gets mapped homeomorphically to $U$, that is, $p|_{U_i} \colon U_i \to U$ is a homeomorphism for all $i$.
\end{definition}

\begin{example}[M\"{o}bius band] \label{ex:mobius}
(This is the continuous analogue of Example~\ref{ex:mobiusDiscrete}, which the reader may wish to recall.)
Consider the map $\C \to \C$ which sends $z \in \C$ to $z^2$. Let $p$ be the restriction of this map to the unit circle $S^1 = \{e^{i \theta} : \theta \in [0,2\pi)\}$; then $p\colon S^1 \to S^1$ is a 2-sheeted covering. Topologically, this is the same as the following construction: A M\"{o}bius band $M$ has a circle $S^1$ running down its center, which we refer to as its equator. Let $p\colon M \to S^1$ be the map which sends a point on $M$ to the nearest point on the equator. Let $D$ be the boundary of the M\"{o}bius band. Then $p|_D\colon D \to S^1$ is a 2-sheeted covering space of the circle, isomorphic to that of Example~\ref{ex:mobiusDiscrete}.
\end{example}

As in the case of graphs, an \emph{isomorphism of covering spaces} $p_\ell\colon Y_\ell \to X$, ($\ell=1,2$) is a homeomorphism $f\colon Y_1 \to Y_2$ such that $p_2 \circ f = p_1$.

For general topological spaces we need one more notion. Given a covering space $p\colon Y \to X$, we say that a collection of open sets $\{U_i \subseteq X :  i \in I\}$ is a \emph{locally trivializing atlas of $X$} if: (1) the union of the $U_i$ equals $X$, and (2) each $U_i$ satisfies the hypothesis in the definition of covering space, that is $p^{-1}(U_i)$ is a disjoint union of open sets $U_{ij}$, and $p|_{U_{ij}}\colon U_{ij} \to U_i$ is a homeomorphism. 

\begin{example} [Graph lifts as covering spaces]
A covering graph of a graph is always a covering space. In any covering graph of a graph $X$, we may take the \emph{stars} of the vertices as a locally trivializing atlas. The star of a vertex $x \in V(X)$ consists of $x$ together with $3/4$ of each edge incident on $x$ (any fraction $> 1/2$ will do; we just need the stars to actually cover all of $X$).
\end{example}

Given a locally trivializing atlas $\{U_i\}$ of a $k$-sheeted covering space $p\colon Y \to X$, label each component of $p^{-1}(U_i)$ with a unique label from $[k]$. If $U_i$ and $U_j$ intersect (nontrivially), then the points of $U_i \cap U_j$ have (at least) two labels assigned to them: one from the labelling of the components of $p^{-1}(U_i)$ with $[k]$, and one from the labelling of the components of $p^{-1}(U_j)$. To each point $u \in U_i \cap U_j$, we may thus assign a permtuation $\pi_{ij}(u) \in S_k$, such that points with label $x$ in $p^{-1}(U_i)$ are the same as points with label $\pi_{ij}(u)(x)$ in $p^{-1}(U_j)$. Clearly, $\pi_{ij}(u) = (\pi_{ji}(u))^{-1}$. These $\pi_{ij}$ are called the \emph{transition functions} of the locally trivializing atlas $\{U_i\}$. 

Note that, since the set $[k]$ is discrete, on each connected component of $U_i \cap U_j$, $\pi_{ij}(u)$ must be constant, so can be represented as a single group element. Without loss of generality, we may refine any finite locally trivializing atlas to another one in which all pairwise intersections are connected. 

\begin{definition}[$G$-covering space]
Let $G$ be a group of permutations of a set of size $k$, and let $X$ be a topological space. A \emph{$G$-covering space} of $X$ is a covering space $p\colon Y \to X$ such that there exists a locally trivializing atlas $\{U_i\}$ of $X$ in which every transition function belongs to the permutation group $G$.
\end{definition}

As with $G$-covering graphs, we consider two $G$-covering spaces to be isomorphic \emph{as $G$-covering spaces} (rather than merely as covering spaces) if there exists a locally trivializing atlas such that the conditions of Definition~\ref{def:isoGcov} hold.

When $k=|G|$ and $G$ acts on itself by (left) translation, we refer to a $G$-covering space with this action as a \emph{principal} $G$-covering space. In Section~\ref{app:nonab_coho} we'll see that, with a suitable (standard) definition of non-Abelian cohomology, \CohLoc with coefficients in an arbitrary group $G$ is equivalent to Maximum Section of Principal $G$-Covering Spaces.

\subsection{Proof of Theorem~\ref{thm:nonab}} \label{app:nonab_proof}
First, recall that in a group $G$, a \emph{commutator} is an element of the form $[g,h] = g^{-1} h^{-1} g h$. The commutator subgroup is the subgroup $G'$ generated by all commutators. In general, there will be elements of $G'$ which are not themselves commutators, but are only products of at least two commutators. $G/G'$ is the maximal quotient of $G$ which is Abelian.

\begin{proof}[Proof of Theorem~\ref{thm:nonab}]
The construction begins exactly as in Proposition~\ref{prop:graphToSurf}, but we run into one additional obstacle. In the Abelian case, the fact that the new cell had each edge in its boundary exactly twice, once in each direction, meant that the sum of the group elements around the boundary was zero. But in the non-Abelian case, all it means is that the product of group elements around the boundary lies in $G'$ (we can see this because when we consider the Abelianization $G/G'$ of $G$, this product again becomes trivial, and therefore it must be an element of $G'$). 

If every element of $G'$ is in fact a commutator (a hypothesis of the theorem), then we can get around this issue as follows. Let $g$ be the product of the group elements around the boundary of the 2-cell (suitably inverted when traversing an edge backwards), starting from some vertex $x$; by assumption, we have $g=[h,k]$ for some $h,k \in G$. Add six new vertices and two new 3-cycles to $X$, say $x \to x_1 \to x_2 \to x$ and $x \to y_1 \to y_2 \to x$. Let the group element on the edges from $x$ to these new vertices be trivial, and let the group element on $(x_1, x_2)$ be $h$ and that on $(y_1, y_2)$ be $k$. Also add new edges from these new vertices to the universal vertex $u$; then as the old graph satisfied the hypothesis of Theorem~\ref{thm:xuong}, so does the new one. Moreover, we can see that what happens is that we enlarge the 2-cell by making its boundary as before, but inserting $x \to x_1 \to x_2 \to x \to y_1 \to y_2 \to x \to x_2 \to x_1 \to x \to y_2 \to y_1 \to x$ between the first and last edge of the cycle (starting at $x$). Moreover, the product around the boundary of this new cell is now trivial by construction, so we indeed have a valid $G$-covering space on a cell decomposition of a 2-manifold.

To get the ``in particular'' for $S_k$, we use a classic result of Ore \cite{ore}, that all elements of $A_k = S_k'$ are commutators.
\end{proof}

\subsection{A brief introduction to (non-Abelian) cohomology, and its equivalence with \texorpdfstring{$G$-covering}{G-covering} spaces} \label{app:nonab_coho}
Here we give a definition of cohomology with non-Abelian coefficient group $G$. This definition is essentially standard (see, e.\,g., \cite{nlab, encycNonab, breen}), but most references we could find for it were quite advanced, so we give the elementary definition here.

Let $G$ be an arbitrary group (not necessarily Abelian). Given a (combinatorial) CW complex $X$, from now on let $X_d$ denote the set of its $d$-simplices (by which we mean the $d$-simplices used in the construction of $X$). A \emph{$d$-cochain} on $X$ with coefficients in $G$ is a function $f\colon X_d \to G$. Two $d$-cochains $f,f'\colon X_1 \to G$ may be multiplied pointwise---$(ff')(s) = f(s) f'(s)$---and it is clear that, with this multiplication, the set of $d$-cochains is a group isomorphic to $G^{|X_d|}$, denoted $C^d(X;G)$. 

A \emph{1-cocycle} on $X$ with coefficients in $G$ is a 1-cochain $f$ which furthermore satisfies the condition that, for every 2-simplex $\Delta \in X_2$, if $e_{i_1}, \dotsc, e_{i_m}$ is the walk in $X_1$ defined by the boundary of $\Delta$ (see Section~\ref{app:complexes}), then 
\[
f(e_{i_1})^{\pm 1}f(e_{i_2})^{\pm 1} \dotsb f(e_{i_m})^{\pm 1} = 1,
\] 
where the exponent on $f(e_{i_j})$ is positive if $e_{i_j}$ is traversed in the direction agreeing with the orientation of that edge, and negative otherwise. The set of 1-cocycles is denoted $Z^1(X ; G)$. If $G$ is Abelian, then again pointwise product makes $Z^1(X ; G)$ into a group, but in the non-Abelian case this will not work in general, and $Z^1(X; G)$ will be treated simply as a set.

The group of 0-cochains $C^0(X ; G)$ acts on $C^1(X ; G)$ in a way that sends 1-cocycles to 1-cocycles, which we'll define next. The key idea is that two 1-cocycles are \emph{cohomologous} if they are in the same orbit under this action. The orbits of this action form the 1-cohomology set $H^1(X ; G)$. When $G$ is Abelian, one can map $C^0(X ; G)$ into a subgroup of $C^1(X ; G)$ (the ``1-coboundaries'') such that the action is simply the group multiplication, and then $H^1(X ; G)$ becomes the quotient group, but when $G$ is not Abelian we are stuck with actions and quotient sets. For related reasons, when $G$ is non-Abelian, only the cohomology groups $H^0$ and $H^1$ really make sense.

The action of $C^0$ on $C^1$ is defined as follows: Given $c \in C^0(X ; G)$, $f \in C^1(X ; G)$, and $e \in X_1$ with endpoints $x,y \in X_0$, such that $e$ is directed from $x$ to $y$, we define
\[
(f^c)(e) = c(x)f(e)c(y)^{-1}.
\]
Note that, in the condition for a 1-cochian to be a cocycle, the action of $c$ does not change the product $f(e_{i_1})^{\pm 1} f(e_{i_2})^{\pm 1} \dotsb f(e_{i_m})^{\pm 1}$, since in between each factor gets inserted a term of the form $c(x)^{-1} c(x)$. Thus $C^0(X ; G)$ also acts on $Z^1(X ; G)$, and the quotient set $H^1(X ; G)$ is well-defined.



Finally, we come to the non-Abelian analogue of the equivalence between $\CohLoc(A)$ and $\GammaMaxLin(A)$ (Observation~\ref{obs:maxLinCoho}). Now that we have all the definitions in place for non-Abelian $G$, the only additional ingredient in the proof is the key idea behind Theorem~F1, 
which we recall now. The basic idea is that, given an open covering $\{U_i\}$ of $X$ and functions $\pi_{ij}\colon U_i \cap U_j \to G$, the necessary and sufficient condition for the $\pi_{ij}$ to be the transition functions of a (principal) $G$-covering space of $X$ is that, whenever there is a triple-overlap $U_i \cap U_j \cap U_k \neq \emptyset$, then 
\begin{equation} \label{eq:coho}
\pi_{ij}(u) \pi_{jk}(u) = \pi_{ik}(u)
\end{equation}
 for all $u \in U_i \cap U_j \cap U_k$. To see what this has to do with cohomology, suppose $X$ is a simplicial complex with the $U_i$ centered at the vertices, and such that $U_i \cap U_j$ is nonempty iff $[i,j]$ is an edge of $X$ and $U_i \cap U_j \cap U_k$ is nonempty iff $[i,j,k]$ is a triangle in $X$. Then $\pi_{ij}$ defines a continuous function on the edge $[i,j]$, and if $G$ is a finite group then this continuous function must be constant, so $\pi_{ij}$ assigns a group element $g_{ij}$ to each edge $[i,j]$. Then the condition (\ref{eq:coho}) is precisely the condition for these group elements form a 1-cocycle. The action of $C^0(X; G)$ on $Z^1(X; G)$ is then given by relabeling the sheets over a given vertex $i \in X_0$, and this action gives (essentially by definition) isomorphisms of $G$-covering spaces (see Observation~\ref{obs:isoGcov} and Definition~\ref{def:isoGcov}). Thus isomorphism classes of principal $G$-covering spaces are in bijective correspondence with the elements of $H^1(X; G)$.


\begin{observation} \label{obs:maxLinCoho_nonab}
Let $G$ be an arbitrary finite group. 1-Cohomology Localization on a space $X$ with coefficients in $G$ is equivalent to Maximum Section of Principal $G$-Covering Spaces on $X$. Indeed, they are essentially the same problem.
\end{observation}

\begin{proof}
Given an instance of Maximum Section of a Principal $G$-Covering Space specified by constants $g_{ij} \in G$ (that is, with constraints $x_i x_j^{-1} = g_{ij}$ for all edges in the cell decomposition of $X$), from the above discussion we see that the $g_{ij}$ define a 1-cocycle on $X$. Let $p\colon Y \to X$ be the corresponding $G$-covering space. As $Y$ is a principal $G$-covering space by definition, we may identify $p^{-1}(x)$ with $G$, for each vertex $x \in X_0$. 

We claim that the following is a natural bijection between assignments $s\colon X_0 \to G$ and 1-cocycles cohomologous to $\{g_{ij}\}$ such that the set of edges to which $s$ can be extended continuously is precisely the complement of the support of the associated 1-cocycle: Given a group element $s_i \in G$ assigned to the vertex $i$, treat $s = \{s_i\}$ as a 0-cochain, and consider the 1-cocycle $g^s$, defined as above by $(g^s)_{ij} = s_i g_{ij} s_j^{-1}$. Note that $s$ can be continuously extended to include an edge $(i,j)$ in its domain if and only if $s_i g_{ij} = s_j$, which is equivalent to $(g^s)_{ij}=1$. 

Conversely, given a cohomologous 1-cocycle $g'$, there is some 0-cochain $s$ such that $g' = g^s$. Treat $s$ as an assignment of a group element $s_i \in G$ to each vertex $i$. Using the same equation as before, we see that the support of $g'$ is precisely the complement of the set of edges to which $s$ can be extended continuously. 

Thus maximizing the number edges in a subset which admits a section is equivalent to minimizing the support of a cohomologous cocycle.
\end{proof}

\bibliographystyle{alphaurl}
\bibliography{UGAT}

\end{document}